\documentclass[journal]{IEEEtran}
\usepackage{tabularx}
\usepackage{longtable}
\usepackage{subfigure}
\usepackage{enumitem}
\usepackage{amsmath,amssymb,amsthm}
\usepackage{mathtools}
\usepackage{graphicx}
\usepackage{cite}
\usepackage{url}
\usepackage{algorithm}
\usepackage{algpseudocode}
\usepackage{booktabs}
\usepackage{tabularx}
\usepackage{booktabs}

\newtheorem{theorem}{Theorem}

\title{\LARGE \bf
Finite-State Decentralized Policy-Based Control With Guaranteed Ground Coverage
}

\author{Hossein Rastgoftar
\thanks{Hossein Rastgoftar is with the Department of Aerospace and Mechanical Engineering, 
        Tucson, AZ 85719, USA
        {\tt\small hrastgoftar@arizona.edu}}%
}

\begin{document}
\title{Distributed Continuous Aerial Surveillance by UAS Swarms Under Formal Mission Specifications}
\maketitle
\begin{abstract}
Persistent aerial surveillance using multi-unmanned aerial systems (UASs) requires decentralized coordination, continuous team reconfiguration, and provable mission correctness despite limited onboard energy and communication constraints. This paper develops a distributed framework for continuous aerial surveillance under bounded Linear Temporal Logic (LTL) mission specifications. The proposed approach partitions the UAS team into stationary anchors and mobile workers operating under cyclic replacement modes, and constructs a deep neural network (DNN)-inspired communication topology that enables fully decentralized coordination through local interactions. A hierarchical bounded LTL specification formally captures  mode-to-mode reference consistency, cyclic team rotation, finite-time reachability, trajectory tracking, and prescribed surveillance coverage. By proving the finite-time convergence of the worker-agent coordination dynamics, the paper gaurantees the  finite-time satisfaction of the mission specification. To maximize sensing effectiveness, an information-theoretic optimization framework synthesizes the reference configuration of newly deployed worker agents by minimizing the Kullback--Leibler divergence between the surveillance-node distribution and the induced coverage density. The resulting reference configuration uniquely determines a deterministic, mode-dependent communication topology, eliminating online communication-graph optimization while preserving the formal mission guarantees. Finally, a decentralized quadrotor controller realizes the distributed references using only local communication. Numerical simulations demonstrate cyclic team reconfiguration, decentralized communication-topology synthesis, finite-time formation convergence, and certified persistent surveillance coverage.

\end{abstract}

\section{Introduction}
Aerial surveillance is a critical component of modern security operations,
providing persistent situational awareness for applications including
border security \cite{darazs2024barrier},
critical infrastructure protection \cite{costin2023bridges},
disaster response \cite{abbas2024disaster},
and emergency management \cite{yucesoy2025drones}. While satellite imagery offers broad-area coverage, its spatial and temporal resolutions are often insufficient for identifying dynamic threats. Low-altitude UASs equipped with high-resolution sensors enable real-time, close-range observation, but a single UAS is constrained by limited sensing range and battery endurance. These limitations motivate cooperative multi-UAS surveillance systems, in which multiple vehicles coordinate to provide persistent, scalable, and resilient monitoring. By dynamically reallocating sensing responsibilities and replacing depleted vehicles without interrupting the mission, such systems significantly improve coverage efficiency, operational robustness, and timely threat detection.

\subsection{Related Work}
Distributed multi-UAV coordination have been extensively investigated over the past decade. Recent advances include distributed time-varying formation and containment tracking \cite{cai2022distributed,gong2023resilient,liu2024distributed}, safety-critical formation control using control barrier functions \cite{liu2026distributed}, resilient event-triggered coordination under cyber and denial-of-service attacks \cite{ren2025resilient}, leader--follower formation control with prescribed performance guarantees \cite{chen2020leader}, and adaptive formation tracking in the presence of actuator faults \cite{wu2023distributed}. Persistent aerial coverage has likewise received significant attention through optimization-, geometry-, and distributed control-based frameworks \cite{liu2021cooperative,chen2021clustering,funada2023distributed,rezaee2024resilient,faghihi2025distributed,zhong2011distributed}. Cooperative path planning and adaptive formation strategies have been developed for multi-UAV surveillance in obstacle-rich environments \cite{liu2021cooperative}, while distributed coverage and data collection have been jointly optimized under sensing, communication, and obstacle constraints \cite{zhong2011distributed}. Distributed control barrier functions have been employed to eliminate coverage holes \cite{funada2023distributed}, resilient Voronoi-based methods have addressed cyberattacks on communication networks \cite{rezaee2024resilient}, and adaptive coverage algorithms have improved deployment in both known and unknown environments \cite{faghihi2025distributed}. Despite these advances, existing coverage frameworks primarily optimize sensing performance and trajectory generation without incorporating finite-time mission specifications, communication-topology synthesis, or formal correctness guarantees for persistent autonomous coverage.

Formal methods have also emerged as an effective tool for specifying, planning, and verifying autonomous multi-agent missions. Mixed-integer optimization has been employed for multi-agent motion planning under Signal Temporal Logic (STL) specifications with collision avoidance and automatic task allocation \cite{sun2022multi}. The MT$^\ast$ algorithm improves the scalability of automata-based planning for multi-robot systems through task decomposition and reduced product-graph construction \cite{gujarathi2022mt}. Team-level spatio-temporal specifications have been extended through t-STREL and differentiable robustness semantics, enabling gradient-based controller synthesis for multi-agent systems \cite{liu2025quantifying}. Zhao \emph{et al.} proposed STL-GO, which augments STL with graph operators to specify behaviors over multiple interaction topologies and supports distributed runtime monitoring \cite{zhao2025stl}. Ma \emph{et al.} introduced Spatial Aggregation Signal Temporal Logic (SaSTL), extending STL with spatial aggregation and counting operators for scalable runtime monitoring of large-scale distributed systems \cite{ma2021novel}, while Zhang \emph{et al.} integrated SaSTL-based monitoring with hierarchical UAV--UGV coverage control and collision-free path planning \cite{zhang2026temporal}. Formal verification techniques based on BDI agents, timed automata, and model checking have also been developed to certify the safety of autonomous vehicle platooning protocols \cite{kamali2017formal}. 
\subsection{Contributions}

Existing multi-UAS coordination and coverage methods provide effective
formation control and coverage optimization, while formal-methods
approaches enable rigorous task specification and verification.
However, existing frameworks do not simultaneously integrate
(i) formal mission specifications for persistent surveillance under
cyclic team reconfiguration and (ii) information-theoretic synthesis of
mode-dependent reference configurations for decentralized execution.
Building upon the DNN-inspired decentralized communication architecture
developed in \cite{10824818}, this paper addresses these limitations
through the following contributions.

\noindent\textbf{Contribution 1: Hybrid Persistent Surveillance Under Formal Mission Specifications:}
We formulate persistent aerial surveillance as a hybrid multi-agent
system operating under deterministic cyclic team reconfiguration and
develop a hierarchical bounded Linear Temporal Logic (LTL)
specification that unifies anchor invariance, mode-to-mode reference
consistency, cyclic team rotation, finite-time reachability,
trajectory tracking, and prescribed surveillance coverage.
The proposed framework establishes a direct connection between formal
mission specifications and decentralized controller design.
Furthermore, by integrating the proposed communication architecture
with the distributed coordination dynamics, the paper establishes
finite-time satisfaction of the prescribed mission specification.

\noindent\textbf{Contribution 2: Information-Theoretic Reference Synthesis for Cyclic Team Replacement:}
We develop an information-theoretic framework for synthesizing the
reference positions of newly introduced worker agents during cyclic
team replacement. The proposed method sequentially minimizes the
Kullback--Leibler divergence between the surveillance-node distribution
and the coverage density induced by the worker-agent team, thereby
placing replacement agents in underrepresented regions of the
surveillance domain. The resulting reference configuration not only
improves surveillance quality but also directly determines the
mode-dependent communication topology required for decentralized
execution.
\subsection{Outline}

The remainder of this paper is organized as follows. Section~\ref{Problem Statement} formulates the surveillance problem and presents the system model and mission specifications. Section~\ref{Structuring Inter-Agent Communication} develops the DNN-based communication architecture. Section~\ref{Network Dynamics and Control} presents the network dynamics and convergence analysis. Section~\ref{Information-Theoretic Coverage Augmentation} develops the information-theoretic reference synthesis framework. Section~\ref{Results} presents simulation results, and Section~\ref{Conclusion} concludes the paper.

\section{Problem Statement}\label{Problem Statement}

Consider an aerial security-surveillance mission performed by a team of
\(N\) UASs indexed by
$
\mathcal V=\{1,\ldots,N\},
$
which is partitioned as
\vspace{-0.25cm}
\begin{equation}
\mathcal V
=
\mathcal A
\dot\cup
\mathcal W,
\end{equation}
where \(\mathcal A\) denotes the anchor-agent set and
\(\mathcal W=\mathcal V\setminus\mathcal A\) denotes the worker-agent set.
The anchor agents are stationary reference nodes with known positions and
continuously broadcast their positions to neighboring worker agents. Let $\mathcal{R}=\{1,\ldots,M\}$ denote the set of operation modes. The directed graph
\vspace{-0.25cm}
\[
\mathcal{G}_{\mathrm{mode}}
=
\left(\mathcal{R},\mathcal{E}_{\mathrm{mode}}\right)
\]
defines the cyclic transitions among the operation modes. For each operation mode
\(\sigma\in\mathcal R\), let
\(\mathcal W_\sigma\subseteq\mathcal W\)
denote the subset of worker agents that preserve their assigned reference
positions throughout mode~\(\sigma\). The corresponding active-agent set is
\vspace{-0.25cm}
\begin{equation}
\mathcal V_\sigma
=
\mathcal A
\dot\cup
\mathcal W_\sigma.
\end{equation}
The remaining worker agents, 
$
\mathcal W\setminus\mathcal V_\sigma
$, 
are reassigned updated reference positions according to the prescribed
rotation rule while retaining their original agent identification numbers.
Consequently, the sets
\(\{\mathcal W_\sigma\}_{\sigma\in\mathcal R}\)
are generally not disjoint and do not form a partition of
\(\mathcal W\). The inter-agent communication is modeled by the directed graph
$
\mathcal{G}_{\sigma}
=
(\mathcal{V}_{\sigma},\mathcal{E}_{\sigma})$, 
for every $\sigma\in\mathcal{R}$,
where
\(
\mathcal{E}_{\sigma}\subseteq
\mathcal{V}_{\sigma}\times\mathcal{V}_{\sigma}
\)
denotes the communication edge set. For each worker agent
\(i\in\mathcal{W}_{\sigma}\), the corresponding in-neighbor set is defined as
\vspace{-0.25cm}
\begin{equation}
\mathcal{N}_{\sigma,i}
=
\left\{
j\in\mathcal{V}_{\sigma}
:
(j,i)\in\mathcal{E}_{\sigma}
\right\},
\qquad
\sigma\in\mathcal{R}.
\end{equation}


The UAS team surveils a region $\mathcal P$, which is uniformly discretized into the surveillance-node set $\mathcal D$. At the beginning of operation mode $\sigma\in\mathcal R$, the surveillance nodes are partitioned into the searched and unsearched sets,
\vspace{-0.25cm}
\begin{equation}
\mathcal D
=
\mathcal S_\sigma
\dot\cup
\mathcal U_\sigma,
\end{equation}
where $\mathcal S_\sigma$ and $\mathcal U_\sigma$ denote the searched and unsearched nodes, respectively. Initially, $\mathcal S_1=\emptyset$ and $\mathcal U_1=\mathcal D$.
Upon completion of operation mode $\sigma$, the searched and unsearched sets are updated according to
\vspace{-0.25cm}
\begin{equation}
\mathcal D
=
\mathcal S_{\sigma^+}
\dot\cup
\mathcal U_{\sigma^+},
\end{equation}
where $\sigma^+$ denotes the subsequent operation mode.

\subsection{Agent Position Representation}

The UAS team operates on a horizontal plane at the prescribed altitude
$z_0$. Each operation mode
$\sigma\in\mathcal R$
is executed over the finite discrete-time interval
\vspace{-0.25cm}
\begin{equation}
\mathcal K_\sigma
=
\{0,\ldots,T_\sigma\},\qquad \sigma\in \mathcal{R},
\end{equation}
where
$T_\sigma\in\mathbb N$ denotes the mode duration. The paper uses the following position notations:

\noindent \textbf{Actual Position:} The actual position of agent $i\in \mathcal{V}$ at time step
$k\in\mathcal K_\sigma$
is
\vspace{-0.25cm}
\begin{equation}
\mathbf r_i(k)
=
\begin{bmatrix}
x_i(k)&
y_i(k)&
z_i(k)
\end{bmatrix}^{\top}.
\end{equation}

\noindent \textbf{Mode Reference Position:} For each operation mode
$\sigma\in\mathcal R$,
define
\vspace{-0.25cm}
\begin{equation}
\mathbf a_{i,\sigma}
=
\begin{bmatrix}
X_{i,\sigma}&
Y_{i,\sigma}&
z_0
\end{bmatrix}^{\top}
\end{equation}
as the reference position of agent $i\in \mathcal{V}$ at the beginning of mode
$\sigma$.
These positions define the reference team configuration for the
corresponding operation mode.

\noindent\textbf{Online Desired Position:}
During operation mode $\sigma\in\mathcal R$, the desired position of worker agent
$i\in\mathcal W_\sigma$ is generated online according to
\vspace{-0.25cm}
\begin{equation}
\mathbf r_{i,d}(k)
=
\sum_{j\in\mathcal N_{\sigma,i}}
w_{ij}^{\sigma}(k)
\mathbf r_j(k),
\qquad
k\in\mathcal K_\sigma,
\end{equation}
where the communication weights satisfy the convexity constraints
\vspace{-0.25cm}
\begin{equation}
\sum_{j\in\mathcal N_{\sigma,i}}
w_{ij}^{\sigma}(k)
=
1,
\qquad
w_{ij}^{\sigma}(k)\ge0,
\end{equation}
for every
$i\in\mathcal W_\sigma$,
$\sigma\in\mathcal R$,
and
$k\in\mathcal K_\sigma$. Let
\vspace{-0.25cm}
\begin{equation}
k_{\sigma}^{\rm mid}
=
\left\lfloor
\frac{T_\sigma}{2}
\right\rfloor
\end{equation}
denote the midpoint of operation mode $\sigma$, where $T_\sigma$ is the
duration of mode $\sigma$. The communication weights are defined as
\vspace{-0.25cm}
\begin{equation}\label{comtweigh}
w_{ij}^{\sigma}(k)
=
\begin{cases}
(1-\beta(k))\varpi_{ij}^{\sigma}
+
\beta(k)\gamma_{ij}^{\sigma},
&
0\le k\le k_{\sigma}^{\rm mid},
\\[2mm]
\gamma_{ij}^{\sigma},
&
k_{\sigma}^{\rm mid}<k\le T_\sigma,
\end{cases}
\end{equation}
where $\varpi_{ij}^{\sigma}$ and $\gamma_{ij}^{\sigma}$ denote the initial
and final communication weights in mode $\sigma$, respectively, and
$\beta(k)$ is a quintic transition function satisfying
\vspace{-0.25cm}
\begin{equation}
\begin{aligned}
&\beta(0)=\dot{\beta}(0)=\ddot{\beta}(0)=0,\\
&\beta(k_{\sigma}^{\rm mid})=1,\qquad
\dot{\beta}(k_{\sigma}^{\rm mid})=
\ddot{\beta}(k_{\sigma}^{\rm mid})=0.
\end{aligned}
\end{equation}
Consequently,
\vspace{-0.25cm}
\begin{equation}
w_{ij}^{\sigma}(k)
=
\gamma_{ij}^\sigma,
\qquad
k\ge k_{\sigma}^{\rm mid},
\end{equation}
for every
$i\in\mathcal W_\sigma$,
$j\in\mathcal N_{\sigma,i}$,
and
$\sigma\in\mathcal R$,
placing the online desired position at the centroid of the neighboring agents. The initial communication weights are uniquely determined from the barycentric-coordinate constraints
\begin{equation}\label{initcomweights}
\sum_{j\in\mathcal N_{\sigma,i}}
\varpi_{ij}^{\sigma}
\mathbf a_{j,\sigma}
=
\mathbf a_{i,\sigma},
\qquad
\sum_{j\in\mathcal N_{\sigma,i}}
\varpi_{ij}^{\sigma}
=
1.
\end{equation}

\noindent \textbf{Terminal Desired Position:} The desired terminal position is defined as the terminal value of the
online desired trajectory,
\begin{equation}
\mathbf p_{i,\sigma}
=
\mathbf r_{i,d}(T),
\qquad
i\in\mathcal V_\sigma.
\end{equation}



\subsection{LTL Mission Specification}\label{LTL}
To formally specify the desired evolution of the aerial surveillance
mission, we employ LTL. The syntax of the LTL formulas considered in this work is defined
recursively as
\begin{equation}
\resizebox{0.99\hsize}{!}{%
$
\varphi
:=
\pi
\mid
\neg\varphi
\mid
\varphi_1\wedge\varphi_2
\mid
\varphi_1\vee\varphi_2
\mid
\varphi_1\rightarrow\varphi_2
\mid
\bigcirc\varphi
\mid
\Box\varphi
\mid
\Diamond_{[0,T]}\varphi,
$
}
\end{equation}
where $\pi$ denotes an atomic proposition,
$\neg$, $\wedge$, $\vee$, and $\rightarrow$
represent logical negation, conjunction, disjunction,
and implication, respectively.
The temporal operators
$\bigcirc$, $\Box$, and
$\Diamond_{[0,T]}$
denote the next, always,
and bounded-eventually operators.
The bounded-eventually operator
$\Diamond_{[0,T]}$
requires the enclosed proposition to become true within the finite
time interval $[0,T]$, which corresponds to the duration of one
operation mode.

The mission specification is constructed from the following atomic
propositions. For each surveillance node
$d\in\mathcal D$,
define
\begin{equation}
\pi_d
\Longleftrightarrow
d\in\mathcal S_\sigma,
\end{equation}
which indicates that surveillance node $d$ has been observed during operation mode $\sigma\in \mathcal{R}$. For each operation mode
$\ell\in\mathcal R$,
define
\begin{equation}
\pi_\sigma^\ell
\Longleftrightarrow
\sigma=\ell,
\end{equation}
which specifies that the aerial surveillance system is operating
under mode $\ell$. For each worker agent
$i\in\mathcal W_\sigma$,
define
\begin{equation}
\pi_{i,\sigma}^{\rm reach}
\Longleftrightarrow
\mathbf r_i
\in
\operatorname{co}
\left\{
\mathbf r_j:
j\in\mathcal N_{\sigma^+,i}
\right\},
\end{equation}
where
$\operatorname{co}(\cdot)$
denotes the convex hull operator.
This proposition states that the worker has entered the convex hull
generated by its prescribed in-neighbors. Finally, define
\begin{equation}\label{pim}
\pi_M
\Longleftrightarrow
\sigma=M,
\end{equation}
which indicates that the surveillance mission is executing its
final operation mode.

The aerial surveillance mission is specified through a collection
of temporal requirements governing the evolution of the agents, the communication topology, cyclic team deployment,
trajectory tracking, and surveillance coverage.
The complete mission specification is obtained as the conjunction
of these individual requirements.

\noindent\textbf{Anchor Invariance:}
The anchor agents remain stationary throughout the mission and continuously
provide reference positions for the workers. This requirement is
specified by
\begin{equation}
\varphi_{\rm anc}
=
\Box
\left(
\bigwedge_{\sigma\in \mathcal{R}}
\bigwedge_{i\in\mathcal A}
\mathbf r_i
=
\mathbf a_{i,\sigma}
\right).
\end{equation}

\noindent\textbf{Mode-to-Mode Reference Consistency:}
Upon completion of each operation mode, the terminal desired position of
every active agent becomes its reference position for the subsequent mode,
thereby guaranteeing continuity of the surveillance mission. This
requirement is specified by
\begin{equation}
\varphi_{\rm ref}
=
\Box
\left[
\bigwedge_{\sigma\in\mathcal R}
\bigwedge_{i\in\mathcal V_\sigma}
\left(
\mathbf p_{i,\sigma}
=
\bigcirc
\mathbf a_{i,\sigma^+}
\right)
\right].
\end{equation}

\noindent\textbf{Cyclic Team Rotation:}
The active team evolves according to the directed cyclic graph
$\mathcal{G}_{\mathrm{mode}}$ so that every team is periodically deployed.
This requirement is specified by
\begin{equation}
\varphi_{\rm rot}
=
\Box
\Bigg(
\bigwedge_{\ell=1}^{M-1}
\Big(
\pi_\sigma^\ell
\rightarrow
\bigcirc
\pi_\sigma^{\ell+1}
\Big)
\wedge
\Big(
\pi_\sigma^M
\rightarrow
\bigcirc
\pi_\sigma^{1}
\Big)
\Bigg).
\end{equation}

\noindent\textbf{Final Reachability:}
For each worker agent
$i\in\mathcal W_\sigma$, define the atomic proposition
\begin{equation}
\pi_{i,\sigma}^{\rm reach}
\Longleftrightarrow
\mathbf r_i
\in
\operatorname{co}
\left\{
\mathbf r_j:
j\in\mathcal N_{\sigma^+,i}
\right\},
\end{equation}
where
$\operatorname{co}(\cdot)$
denotes the convex hull operator.
During each operation mode, every worker agent is required to reach the
convex hull generated by its three in-neighbors before the mode
terminates. This guarantees that the desired communication geometry is
established prior to the subsequent mode transition. Accordingly,
\begin{equation}
\varphi_{\rm reach}
=
\Box
\left[
\bigwedge_{\ell\in\mathcal R}
\left(
\pi_\sigma^\ell
\rightarrow
\bigwedge_{i\in\mathcal W_\ell}
\Diamond_{[0,T]}
\pi_{i,\ell}^{\rm reach}
\right)
\right].
\end{equation}

\noindent\textbf{Tracking Requirement:}
During each operation mode, every worker agent must reach its online
desired position within a prescribed tolerance
$\varepsilon>0$ before the mode terminates. This requirement couples the
communication topology with the trajectory tracking controller and
ensures that the desired formation is realized before the subsequent mode
transition. Therefore,
\begin{equation}
\varphi_{\rm track}
=
\Box
\left[
\bigwedge_{\sigma\in\mathcal R}
\bigwedge_{i\in\mathcal W_\sigma}
\Diamond_{[0,T]}
\left(
\|
\mathbf r_i
-
\mathbf r_{i,d}
\|
\le
\varepsilon
\right)
\right].
\end{equation}


\noindent\textbf{Communication-Topology Requirement:} For each worker agent $i\in\mathcal W_\sigma$, define the atomic proposition
\begin{equation}\label{picom}
\pi_{i,\sigma}^{\rm comm}
\Longleftrightarrow
|\mathcal N_{\sigma,i}|=3,
\end{equation}
which specifies that worker agent $i$ communicates with exactly three
prescribed in-neighbor agents during operation mode $\sigma$.
This guarantees that the
layered DNN communication structure required for decentralized reference
generation is preserved throughout the surveillance mission. Accordingly,
\begin{equation}
\varphi_{\rm comm}
=
\Box
\left[
\bigwedge_{\sigma\in\mathcal R}
\bigwedge_{i\in\mathcal W_\sigma}
\pi_{i,\sigma}^{\rm comm}
\right].
\end{equation}

\noindent\textbf{Coverage Requirement:}
Assume that each UAS is equipped with a downward-looking conical sensor whose
ground footprint is a circular disk of radius $r>0$.
Let
$
\mathcal S_{\sigma^+}\subseteq\mathcal D
$
denote the set of surveillance nodes observed by at least one
UAS upon completion of mode $\sigma$. The corresponding coverage ratio is
\begin{equation}
\eta_{\sigma^+}
=
\frac{|\mathcal S_{\sigma^+}|}{|\mathcal D|},
\end{equation}
which represents the fraction of surveillance nodes that have been covered.
Given a prescribed coverage threshold $x\in(0,1]$, the mission is required to achieve
at least $100x\%$ coverage upon completion of the final mode. This requirement
is expressed as
\begin{equation}
\varphi_{\rm cov}
=
\Box
\left[
\pi_M
\rightarrow
\Diamond_{[0,T]}
\left(
\eta_{M^+}\ge x
\right)
\right],
\end{equation}
where $\pi_M$, defined in \eqref{pim}, indicates that the system is operating in the final mode.

\noindent\textbf{Overall Mission Specification:}
The aerial surveillance mission is successfully accomplished if and only
if all temporal requirements above are satisfied simultaneously.
Accordingly, the overall mission specification is
\begin{equation}
\varphi
=
\varphi_{\rm anc}
\wedge
\varphi_{\rm ref}
\wedge
\varphi_{\rm rot}
\wedge
\varphi_{\rm reach}
\wedge
\varphi_{\rm track}
\wedge\varphi_{\rm comm}\wedge
\varphi_{\rm cov}.
\end{equation}
\noindent\textbf{Mode Selection Requirement:}
The number of operation modes $M$ is selected during the mission-design stage so that one complete surveillance cycle satisfies the prescribed coverage requirement. Specifically, for a desired coverage threshold $x\in(0,1]$, the number of modes is chosen as the smallest positive integer satisfying
\begin{equation}
M
=
\min
\left\{
m\in\mathbb N:
\eta_{m^+}\ge x
\right\},
\end{equation}
where $\eta_{m^+}$ denotes the cumulative coverage ratio achieved upon completion of the $m$th operation mode. Consequently, a complete cycle of the cyclic team-rotation schedule guarantees that at least $100x\%$ of the surveillance nodes have been observed.

\subsection{Objectives of the Paper}
The main goal of this paper is to build a low-level network control model so that the overall mission constraints and requirements, specified by $\varphi$, are met and as the result, the UAS team operation can be deterministically planned by cyclic graph $\mathcal G_\sigma
=
(\mathcal V_\sigma,\mathcal E_\sigma)$ with the properties specified in Section \ref{LTL}. To this end, we solve the following problems:

\noindent\textbf{Problem 1: Structuring Inter-Agent Communication:}
For each operation mode $\sigma\in\mathcal R$, construct a layered
communication topology directly from the arbitrary reference
configuration
$\{\mathbf a_{i,\sigma}\}_{i\in\mathcal V_\sigma}$
using the DNN-inspired communication synthesis algorithm developed in
\cite{10824818}. Unlike the original framework, where the communication
topology is generated for a single reference configuration, the proposed
surveillance framework repeatedly synthesizes mode-dependent
communication graphs corresponding to the evolving reference
configurations induced by cyclic team reconfiguration. The resulting
communication topology ensures that every worker agent
$i\in\mathcal W_\sigma$ communicates with exactly three prescribed
in-neighbor agents,
$\mathcal N_{\sigma,i}$,
thereby enabling fully decentralized reference generation during every
operation mode. 

\noindent\textbf{Problem 2: Decentralized Convergence Guarantee:}
Develop a decentralized control strategy for each worker agent \(i\in\mathcal{W}_{\sigma}\) that guarantees convergence to a desired position \(\mathbf p_{i,\sigma}\), i.e.,
\begin{equation}
\lim_{t\to\infty}\mathbf r_i(t)
=
\mathbf p_{i,\sigma},
\end{equation}
where \(\mathbf p_{i,\sigma}\) is not known a priori to agent  \(i\in\mathcal{W}_{\sigma}\) and satisfies
\begin{equation}
\mathbf p_{i,\sigma}
\in
\operatorname{co}
\left\{
\mathbf a_{j,\sigma^+} :
j\in\mathcal N_{\sigma^{+},i}
\right\}.
\end{equation}
The control law must rely solely on locally available information exchanged through the communication network.

\noindent\textbf{Problem 3: Information-Theoretic Coverage Augmentation:}
Given the final desired positions associated with the worker-agent set
\(\mathcal{W}_{\sigma}\) and the surveillance-node set \(\mathcal{D}\), determine the reference positions of the augmented worker agents
\begin{equation}
\mathcal{W}_{\sigma^{+}}\setminus\mathcal{W}_{\sigma}
\end{equation}
such that the coverage density generated by the augmented team
\(\mathcal{W}_{\sigma^{+}}\)
minimizes an information-theoretic distance measure with respect to the target density. The placement strategy shall maximize coverage in regions that are insufficiently represented by the current team while guaranteeing satisfaction of the coverage specification \(\varphi_{\rm cov}\).

\begin{algorithm}[t]
\caption{Construction of the DNN Communication Topology for Mode $\sigma$}
\label{alg:DNN}
\begin{algorithmic}[1]

\Require Reference agent positions
$\{\mathbf a_{i,\sigma}\}_{i\in\mathcal V_\sigma}$

\Ensure
Worker groups
$\mathcal V_{\sigma,1},\ldots,\mathcal V_{\sigma,m_\sigma}$,
layers
$\mathcal L_{\sigma,0},\ldots,\mathcal L_{\sigma,m_\sigma}$,
and neighbor sets
$\{\mathcal N_{\sigma,i}\}$.

\State Determine the boundary-agent set
$\mathcal V_{B,\sigma}$.

\State Compute the core anchor agent
$b_c$ using (16).

\State Define the anchor set $
\mathcal A
=
\mathcal V_{B,\sigma}\cup\{b_c\}$.
\State Set $\mathcal V_{\sigma,0}
=
\mathcal L_{\sigma,0}
=
\mathcal A,
\qquad
\bar{\mathcal L}_{\sigma,0}
=
\mathcal V_\sigma\setminus\mathcal A$.
\State Triangulate
$\mathrm{co}(\mathcal A)$:~
$
\mathcal T_{\sigma,0}
=
\{
\mathcal R_{\sigma,0,1},
\dots,
\mathcal R_{\sigma,0,m_{\sigma,0}}
\}.
$

\State $l\gets1$

\While{$\bar{\mathcal L}_{\sigma,l-1}\neq\emptyset$}

    \State
    $\mathcal V_{\sigma,l}\gets\emptyset$

    \State
    $\mathcal T_{\sigma,l}\gets\emptyset$

    \For{each triangular cell
    $\mathcal R_{\sigma,l-1,q}\in\mathcal T_{\sigma,l-1}$}

        \If{$\mathrm{co}(\mathcal R_{\sigma,l-1,q})$
        contains at least one worker}
            \State Select the enclosed worker
            \[
            i^\star
            =
            \arg\min_i
            \sum_{j\in\mathcal R_{\sigma,l-1,q}}
            \|
            \mathbf a_{i,\sigma}
            -
            \mathbf a_{j,\sigma}
            \|.
            \]
            \State
            $\mathcal N_{\sigma,i^\star}
            \gets
            \mathcal R_{\sigma,l-1,q}$

            \State
            $\mathcal V_{\sigma,l}
            \gets
            \mathcal V_{\sigma,l}
            \cup
            \{i^\star\}$

            \State Subdivide
            $\mathcal R_{\sigma,l-1,q}$
            into three triangular cells
            \State using
            $i^\star$, and append them to
            $\mathcal T_{\sigma,l}$.

        \EndIf

    \EndFor

    \State
    $\mathcal L_{\sigma,l}
    \gets
    \mathcal L_{\sigma,l-1}
    \dot\cup
    \mathcal V_{\sigma,l}$

    \State
    $\bar{\mathcal L}_{\sigma,l}
    \gets
    \mathcal V_\sigma
    \setminus
    \mathcal L_{\sigma,l}$

    \State
    $l\gets l+1$

\EndWhile

\State
$m_\sigma\gets l-1$

\end{algorithmic}
\end{algorithm}

\section{Structuring Inter-Agent Communication}\label{Structuring Inter-Agent Communication}

For each mode $\sigma\in\mathcal R$, $\mathcal{V}_\sigma$
is partitioned into $m_\sigma+1$ disjoint groups,
\begin{equation}
\mathcal V_\sigma
=
\dot{\bigcup}_{l=0}^{m_\sigma}
\mathcal V_{\sigma,l},
\qquad
\mathcal V_{\sigma,l}
\cap
\mathcal V_{\sigma,h}
=
\emptyset,
\quad
l\neq h,
\end{equation}
where
\begin{equation}
\mathcal V_{\sigma,0}
=
\mathcal A.
\end{equation}
These groups induce the DNN layers
$\mathcal L_{\sigma,0},\ldots,\mathcal L_{\sigma,m_\sigma}$,
which are recursively defined by
\begin{equation}
\mathcal L_{\sigma,l}
=
\begin{cases}
\mathcal A,
&
l=0,
\\
\mathcal L_{\sigma,l-1}
\dot\cup
\mathcal V_{\sigma,l},
&
l=1,\ldots,m_\sigma.
\end{cases}
\end{equation}

\subsection{Anchor-Agent Layer}

Let $\mathcal V_{B,\sigma}$ denote the agents lying on the
boundary of the convex hull generated by the reference
configuration
\(
\{\mathbf a_{i,\sigma}\}_{i\in\mathcal V_\sigma}.
\)
A core anchor agent is selected as
\begin{equation}
b_c
=
\arg\min_{i\in
\mathcal V_\sigma\setminus
\mathcal V_{B,\sigma}}
\;
\sum_{j\in\mathcal V_{B,\sigma}}
\|
\mathbf a_{i,\sigma}
-
\mathbf a_{j,\sigma}
\|.
\end{equation}
The anchor set is defined by
\begin{equation}
\mathcal A
=
\mathcal V_{B,\sigma}
\cup
\{b_c\},
\qquad
\sigma\in\mathcal R.
\end{equation}
Consequently,
\begin{equation}
\mathcal L_{\sigma,0}
=
\mathcal A.
\end{equation}

The convex hull generated by the anchor agents is decomposed
into simplices, forming the initial geometric scaffold used to
construct the communication topology.

\subsection{Worker-Agent Layers}\label{Worker-Agent Layers}

The remaining communication layers are generated recursively
from the reference configuration using the triangular-cell
expansion procedure presented by Algorithm \ref{alg:DNN}
\cite{10824818}. Let $
\mathcal T_{\sigma,l-1}
=
\{
\mathcal R_{\sigma,l-1,1},
\ldots,
\mathcal R_{\sigma,l-1,m_{\sigma,l-1}}
\}
$
denote the collection of triangular cells induced by the agents
belonging to layer $\mathcal L_{\sigma,l-1}$. For each triangular cell
$\mathcal R_{\sigma,l-1,q}$,
$q=1,\ldots,m_{\sigma,l-1}$,
one enclosed worker agent, when such an agent exists, is assigned
to the next layer. The selected agent minimizes the aggregate
distance to the three vertices of the corresponding triangular
cell. Repeating this procedure recursively generates the layers
$
\mathcal L_{\sigma,1},
\mathcal L_{\sigma,2},
\ldots,
\mathcal L_{\sigma,m_\sigma}$.

The resulting communication graph possesses a strictly layered
feedforward structure satisfying
\begin{equation}\label{strictlylayered}
\mathcal N_{\sigma,i}
\subseteq
\mathcal L_{\sigma,l-1},
\qquad
\forall i\in\mathcal L_{\sigma,l},
\quad
l\ge 1.
\end{equation}

Per \eqref{picom}, every worker agent communicates with exactly three
agents from the preceding layer, forming an enclosing triangle
in the reference configuration. This construction uniquely
defines the neighbor sets $\mathcal N_{\sigma,i}$ used in the
barycentric update law
\begin{equation}
\mathbf r_{i,d}(k)
=
\sum_{j\in\mathcal N_{\sigma,i}}
w_{ij}^{\sigma}(k)\,
\mathbf r_j(k),
\end{equation}
and therefore provides the communication topology underlying
the LTL mission specification.

For illustration, consider a team of $N=13$ agents forming the two-dimensional configuration shown in Fig.~\ref{CellDecomposition}(a)--(d), where $\mathcal{V}_B=\{1,\dots,4\}$ ($N_B=4$) denotes the boundary agents. Agent $b_c=12$ is selected as the core agent, yielding $\mathcal{V}_{\sigma,0}=\mathcal{L}_{\sigma,0}=\{1,\dots,4,12\}$ for the first layer. The convex hull generated by $\mathcal{L}_{\sigma,0}$ is decomposed into $m_0=4$ triangular cells defined by $\mathcal{R}_{0,1}=\{1,2,12\}$, $\mathcal{R}_{0,2}=\{2,3,12\}$, $\mathcal{R}_{0,3}=\{3,4,12\}$, and $\mathcal{R}_{0,4}=\{4,1,12\}$. Since 
$
\mathrm{co}\{\mathbf{q}_i:i\in \mathcal{R}_{0,h}\}, 
$
$h=1,\cdots,4$, contains at least one agent, the decomposition generates $m_1=12$ triangular cells (blue) and the layer $\mathcal{V}_{\sigma,1}=\{11,10,6,5\}$. Among the resulting regions,
$
\mathrm{co}\{\mathbf{a}_i:i\in \mathcal{R}_{1,h}\}
$
 only those with $h=3,7,10,11$ contain a single agent, yielding the second-layer agents $\mathcal{V}_{\sigma,2}=\{7,8,9,13\}$. The resulting DNN defines the inter-agent communication topology.
\begin{figure}[ht]
\center
\includegraphics[width=3.3 in]{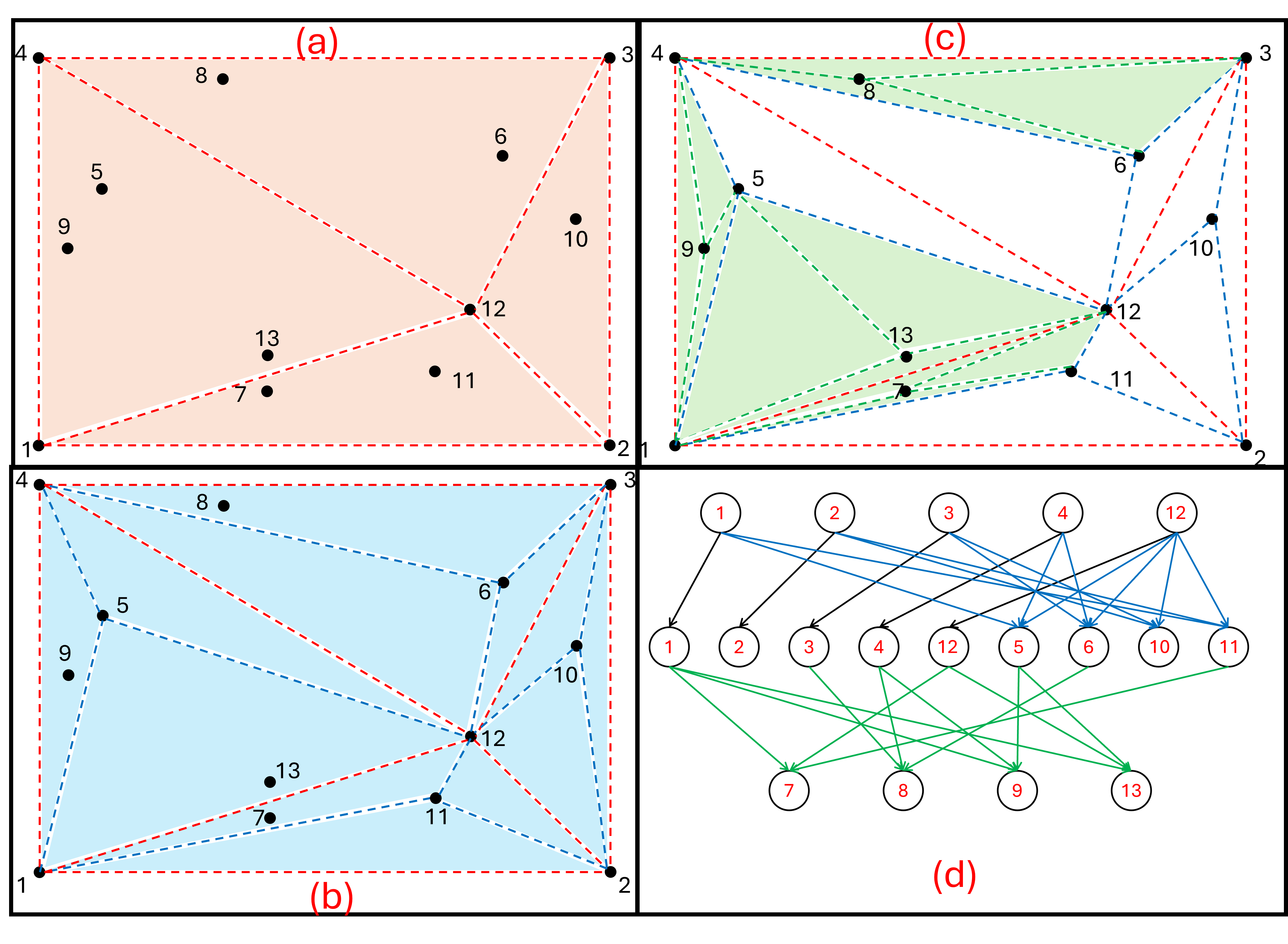}
\vspace{-.4cm}
\caption{
(a) The convex region $CONV(\mathcal{R}_{0,h})$, for $h=1,\dots,4$, includes at least one agent, yielding $m_1=12$ triangular cells (blue), each defined by three agents which result in $\mathcal{V}_{\sigma,1} = \{11, 10, 6, 5\}$.
(b) Among these, $CONV(\mathcal{R}_{1,h})$ contains a single agent for $h=3,7,10,11$; the rest are empty.
(c) At layer 2, $\mathcal{V}_{\sigma,2} = \{7,8,9,13\}$ are assigned as mentees based on neighboring cells.
(d) The resulting DNN determines the communication links between agents \cite{rastgoftar2026finite}. 
}
\label{CellDecomposition}
\end{figure}

\subsection{Transient Communication Weights}

For each worker agent $i\in\mathcal W_\sigma$, the communication weights are
computed according to \eqref{comtweigh} using the initial communication weights
$\varpi_{ij}^\sigma$, the target communication weights
$\gamma_{ij}^\sigma$, and the mode duration $T_\sigma$, where
$\sigma\in\mathcal{R}$. The initial communication weights are determined
from the reference configuration of the active agents using
\eqref{initcomweights}.
The final communication weights are obtained from the spatial distribution of
the unsearched targets. Let each unsearched target
\(\mathbf d_\ell\in\mathcal U_\sigma\) define a normal distribution
\[
\rho_\ell(\mathbf q)
=
\mathcal N(\mathbf q;\mathbf d_\ell,\Sigma_\ell),
\]
where \(\mathbf d_\ell\) is the mean and \(\Sigma_\ell\) is the covariance
matrix. For each worker agent \(i\in\mathcal W_\sigma\), define the desired
coverage region
\[
\mathcal H_{\sigma,i}
=
\mathrm{co}\{\mathbf p_j:j\in\mathcal N_{\sigma,i}\}.
\]
The probability mass of the unsearched target distributions contained in
\(\mathcal H_{\sigma,i}\) is given by
\begin{equation}
J_{\sigma,i}
=
\sum_{\mathbf d_\ell\in\mathcal U_\sigma}
\int_{\mathcal H_{\sigma,i}}
\rho_\ell(\mathbf q)\,d\mathbf q .
\end{equation}
The desired local convex hull is selected to maximize the enclosed probability
mass of the unsearched target distributions:
\begin{equation}
\mathcal H_{\sigma,i}^\star
=
\arg\max_{\mathcal H_{\sigma,i}}
J_{\sigma,i}.
\end{equation}
The final desired position \(\mathbf p_{\sigma,i}\) is then defined as the
centroid of the selected convex hull \(\mathcal H_{\sigma,i}^\star\), i.e.,
\begin{equation}
\mathbf p_{\sigma,i}
=
\frac{1}{|\mathcal H_{\sigma,i}^\star|}
\int_{\mathcal H_{\sigma,i}^\star}
\mathbf q\,d\mathbf q,
\end{equation}
where \(|\mathcal H_{\sigma,i}^\star|\) denotes the area of
\(\mathcal H_{\sigma,i}^\star\). Thus, \(\mathbf p_{\sigma,i}\) is determined
by the geometry of the probability-informed convex hull, rather than by the
centroid of the mean locations of the normal distributions. The final communication weights \(\gamma_{ij}^{\sigma}\) are then obtained by
expressing \(\mathbf p_i\) as a convex combination of the neighboring target
positions:
\begin{equation}
    \mathbf p_i
    =
    \sum_{j\in\mathcal N_{\sigma,i}}
    \gamma_{ij}^{\sigma}\mathbf p_j,
    \qquad
    \sigma\in\mathcal R,~ i\in\mathcal W_\sigma,
\end{equation}
subject to
\begin{equation}
    \sum_{j\in\mathcal N_{\sigma,i}}\gamma_{ij}^{\sigma}=1,
    \qquad
    \gamma_{ij}^{\sigma}\ge 0,
    \qquad
    j\in\mathcal N_{\sigma,i}.
\end{equation}
Thus, the transient communication weights continuously deform the local
barycentric representation from the initial reference geometry to a
coverage-driven final geometry determined by the distribution of the remaining
unsearched targets. 


\begin{theorem}
\label{thm:recursive_anchor_coordinates}

Consider a mode $\sigma\in\mathcal R$ over the finite time interval
$\mathcal K_\sigma=\{0,\ldots,T_\sigma\}$, and suppose that
\eqref{picom} and \eqref{strictlylayered} are satisfied.
Assume that the anchor agents are affinely independent and that the
finite-time tracking specification
\[
\varphi_{\rm track}^{\sigma}
=
\bigwedge_{i\in\mathcal W_\sigma}
\Diamond_{[0,T]}
\left(
\left\|
\mathbf r_i(k)-\mathbf r_{i,d}(k)
\right\|
\le \epsilon
\right)
\]
holds for some $\epsilon\ge0$. Then every desired position admits a
recursive barycentric representation with respect to the anchor agents:
\[
\mathbf r_{i,d}(k)
=
\sum_{a\in\mathcal A}
\lambda_{ia}^{\sigma}(k)\,
\mathbf r_a(k),
\qquad
i\in\mathcal W_\sigma .
\]
The anchor coordinates are initialized as
\[
\lambda_{aa}^{\sigma}=1,
\qquad
\lambda_{ab}^{\sigma}=0,
\quad a\neq b,
\qquad a,b\in\mathcal A,
\]
and are propagated recursively according to
\begin{equation}\label{coreq}
\lambda_{ia}^{\sigma}(k)
=
\sum_{j\in\mathcal N_{\sigma,i}}
w_{ij}^{\sigma}(k)\,
\lambda_{ja}^{\sigma}(k),
\qquad
a\in\mathcal A .
\end{equation}
Moreover,
\[
\lambda_{ia}^{\sigma}(k)\ge0,
\qquad
\sum_{a\in\mathcal A}
\lambda_{ia}^{\sigma}(k)=1.
\]
Consequently, every actual worker position satisfies
\[
\bigwedge_{i\in\mathcal W_\sigma}
\Diamond_{[0,T]}
\left(
\left\|
\mathbf r_i(k)
-
\sum_{a\in\mathcal A}
\lambda_{ia}^{\sigma}(k)\,
\mathbf r_a(k)
\right\|
\le \epsilon
\right).
\]
Furthermore, if the transient communication weights converge to the
constants
\[
\lim_{k\rightarrow\infty}
w_{ij}^{\sigma}(k)
=
\gamma_{ij}^{\sigma},
\qquad
j\in\mathcal N_{\sigma,i},
\]
where
\[
\gamma_{ij}^{\sigma}\ge0,
\qquad
\sum_{j\in\mathcal N_{\sigma,i}}
\gamma_{ij}^{\sigma}=1,
\]
then the anchor coordinates satisfy the limiting recursion
\[
\lambda_{ia}^{\sigma}
=
\sum_{j\in\mathcal N_{\sigma,i}}
\gamma_{ij}^{\sigma}\,
\lambda_{ja}^{\sigma},
\qquad
a\in\mathcal A .
\]
\end{theorem}

\begin{proof}
For the anchor layer,
$\mathcal L_{\sigma,0}=\mathcal A$,
the result follows immediately from $\lambda_{aa}^{\sigma}=1$ and $\lambda_{ab}^{\sigma}=0$ for $a\neq b$.
Assume that every agent
$j\in\mathcal L_{\sigma,l-1}$
admits the representation
\[
\mathbf r_j(k)
=
\sum_{a\in\mathcal A}
\lambda_{ja}^{\sigma}(k)\,
\mathbf r_a(k).
\]
Since
$\mathcal N_{\sigma,i}\subseteq\mathcal L_{\sigma,l-1}$,
every in-neighbor of
$i\in\mathcal L_{\sigma,l}$
satisfies the induction hypothesis. Using the barycentric update law,
\[
\begin{split}
\mathbf r_{i,d}(k)
&=
\sum_{j\in\mathcal N_{\sigma,i}}
w_{ij}^{\sigma}(k)\,
\mathbf r_j(k)\\
&=
\sum_{j\in\mathcal N_{\sigma,i}}
w_{ij}^{\sigma}(k)
\sum_{a\in\mathcal A}
\lambda_{ja}^{\sigma}(k)\,
\mathbf r_a(k)\\=&
\sum_{a\in\mathcal A}
\left(
\sum_{j\in\mathcal N_{\sigma,i}}
w_{ij}^{\sigma}(k)
\lambda_{ja}^{\sigma}(k)
\right)
\mathbf r_a(k).
\end{split}
\]
Hence, \eqref{coreq} follows, proving the recursive representation. Since
$w_{ij}^{\sigma}(k)\ge0$
and
$\lambda_{ja}^{\sigma}(k)\ge0$,
it follows that
$\lambda_{ia}^{\sigma}(k)\ge0$.
Furthermore,
\vspace{-0.25cm}
\[
\begin{split}
\sum_{a\in\mathcal A}
\lambda_{ia}^{\sigma}(k)=
\sum_{j\in\mathcal N_{\sigma,i}}
w_{ij}^{\sigma}(k)
\sum_{a\in\mathcal A}
\lambda_{ja}^{\sigma}(k)=
\sum_{j\in\mathcal N_{\sigma,i}}
w_{ij}^{\sigma}(k)
=1,
\end{split}
\]
which proves convexity. By the finite-time tracking specification
$\varphi_{\rm track}^{\sigma}$, for every
$i\in\mathcal W_\sigma$ there exists
$k_i^\star\in\{0,\ldots,T\}$ such that
\vspace{-0.25cm}
\begin{equation}
\left\|
\mathbf r_i(k_i^\star)
-
\mathbf r_{i,d}(k_i^\star)
\right\|
\le \epsilon .
\end{equation}
Since
\vspace{-0.25cm}
\begin{equation}
\mathbf r_{i,d}(k_i^\star)
=
\sum_{a\in\mathcal A}
\lambda_{ia}^{\sigma}(k_i^\star)\,
\mathbf r_a(k_i^\star),
\end{equation}
it follows that
\vspace{-0.25cm}
\begin{equation}
\left\|
\mathbf r_i(k_i^\star)
-
\sum_{a\in\mathcal A}
\lambda_{ia}^{\sigma}(k_i^\star)\,
\mathbf r_a(k_i^\star)
\right\|
\le \epsilon .
\end{equation}
Equivalently,
\vspace{-0.25cm}
\begin{equation}
\bigwedge_{i\in\mathcal W_\sigma}
\Diamond_{[0,T]}
\left(
\left\|
\mathbf r_i(k)
-
\sum_{a\in\mathcal A}
\lambda_{ia}^{\sigma}(k)\,
\mathbf r_a(k)
\right\|
\le \epsilon
\right)
\end{equation}
holds. Finally, if
\[
\lim_{k\rightarrow\infty}
w_{ij}^{\sigma}(k)
=
\gamma_{ij}^{\sigma},
\qquad
j\in\mathcal N_{\sigma,i},
\]
then taking the limit of the recursion
\eqref{coreq} yields
\[
\lambda_{ia}^{\sigma}
=
\sum_{j\in\mathcal N_{\sigma,i}}
\gamma_{ij}^{\sigma}\,
\lambda_{ja}^{\sigma},
\qquad
a\in\mathcal A,
\]
which establishes the limiting anchor-coordinate recursion. This
completes the proof.
\end{proof}

Let $\mathcal A=\{a_1,\ldots,a_L\}$ and $W_\sigma=\{b_1,\ldots,b_{N_\sigma}\}$. Then,
\vspace{-0.25cm}
\begin{equation}
\mathbf W_\sigma(k)
=
\begin{bmatrix}
\mathbf B_\sigma(k) &
\mathbf A_\sigma(k)
\end{bmatrix}
\in
\mathbb R^{N_\sigma\times(L+N_\sigma)}
\end{equation}
defines the weighted communication matrix that satisfies:
\vspace{-0.25cm}
\begin{equation}
\left(\mathbf W_\sigma(k)\otimes \mathbf I_3\right)
\begin{bmatrix}
\mathbf r_{\mathcal A}(k)\\
\mathbf r_{\mathcal W_\sigma}(k)
\end{bmatrix}
=
\mathbf 0,
\end{equation}
where
\vspace{-0.25cm}
\begin{equation}
\mathbf r_{\mathcal A}
=
\begin{bmatrix}
\mathbf r_{a_1}^\top &
\cdots &
\mathbf r_{a_L}^\top
\end{bmatrix}^{\!\top},
\qquad
\mathbf r_{\mathcal W_\sigma}
=
\begin{bmatrix}
\mathbf r_{b_1}^\top &
\cdots &
\mathbf r_{b_{N_\sigma}}^\top
\end{bmatrix}^{\!\top}.
\end{equation}
The entries of $\mathbf W_\sigma(k)$ are
\vspace{-0.25cm}
\begin{equation}
[\mathbf W_\sigma(k)]_{ij}
=
\begin{cases}
w_{b_i a_j}^{\sigma}(k),
&
j\le L,\; a_j\in\mathcal N_{\sigma,b_i},
\\[1mm]
-1,
&
j=L+i,
\\[1mm]
w_{b_i b_{j-L}}^{\sigma}(k),
&
j>L,\; b_{j-L}\in\mathcal N_{\sigma,b_i},
\\[1mm]
0,
&
\text{otherwise}.
\end{cases}
\end{equation}

Since $w_{ij}^{\sigma}(k)$ is constant at ever $k\geq k_\sigma^{\rm mid}$,
\vspace{-0.25cm}
 \begin{equation}
\mathbf A_\sigma
=
\mathbf A_\sigma(k_\sigma^{\rm mid}),
\qquad
\mathbf B_\sigma
=
\mathbf B_\sigma(k_\sigma^{\rm mid}),
\end{equation}
are constant matrices.

\vspace{-0.25cm}
\begin{theorem}
\label{thm:Hmatrix}
Suppose the inter-agent communication graph is generated by the DNN
construction of Section~\ref{Worker-Agent Layers}. Then:
\begin{enumerate}
\item
There exists a permutation matrix
$\mathbf P_\sigma$
such that
\vspace{-0.25cm}
\[
\tilde{\mathbf A}_\sigma
=
\mathbf P_\sigma
\mathbf A_\sigma
\mathbf P_\sigma^\top
\]
is lower triangular with
\vspace{-0.25cm}
\[
[\tilde{\mathbf A}_\sigma]_{ii}
=
-1,
\qquad
i=1,\ldots,N_\sigma .
\]
Consequently, $\mathbf A_\sigma$ is nonsingular and Hurwitz.
\item
The matrix
\vspace{-0.25cm}
\begin{equation}
\mathbf H_\sigma
=
-\mathbf A_\sigma^{-1}\mathbf B_\sigma
\end{equation}
is nonnegative and row stochastic, i.e.,
\begin{equation}
\mathbf H_\sigma\ge0,
\qquad
\mathbf H_\sigma\mathbf 1_L
=
\mathbf 1_{N_\sigma}.
\end{equation}

\item
The $i$th row of $\mathbf H_\sigma$ contains the final
anchor-coordinate weights of worker agent $b_i$.
\end{enumerate}
\end{theorem}
\vspace{-0.25cm}
\begin{proof}
Because the DNN communication graph is strictly layered, every worker
agent communicates only with anchor agents or worker agents belonging to
preceding layers. Consequently, if
$b_j\in\mathcal N_{\sigma,b_i}$, then necessarily $j<i$ after a suitable
layer-consistent ordering of the worker agents. Therefore,
$\mathbf A_\sigma$ is permutation similar to a lower triangular matrix
with diagonal entries equal to $-1$. Its eigenvalues are all equal to
$-1$, implying that $\mathbf A_\sigma$ is nonsingular and Hurwitz.

At the final communication weights, the barycentric relations of all
worker agents can be written compactly as
\vspace{-0.25cm}
\begin{equation}
\left(\mathbf B_\sigma\otimes \mathbf I_3\right)
\mathbf r_{\mathcal A}
+
\left(\mathbf A_\sigma\otimes \mathbf I_3\right)
\mathbf r_{\mathcal W_\sigma}
=
\mathbf 0,
\end{equation}
where $\mathbf r_{\mathcal A}
\in \mathbb R^{3L}$ and $\mathbf r_{\mathcal W_\sigma}
\in \mathbb R^{3N_\sigma}$
are the stacked three-dimensional anchor and worker position vectors.
Since $\mathbf A_\sigma$ is nonsingular,
\vspace{-0.25cm}
\begin{equation}
\mathbf r_{\mathcal W_\sigma}
=
-
\left(\mathbf A_\sigma^{-1}\mathbf B_\sigma\otimes \mathbf I_3\right)
\mathbf r_{\mathcal A}
=
\left(\mathbf H_\sigma\otimes \mathbf I_3\right)
\mathbf r_{\mathcal A}.
\end{equation}
By Theorem~\ref{thm:recursive_anchor_coordinates}, each worker position is
a convex combination of the anchor positions. Hence, every
anchor-coordinate weight is nonnegative and the weights associated with
each worker sum to one. Therefore,
\vspace{-0.25cm}
\[
\mathbf H_\sigma\ge0,
\qquad
\mathbf H_\sigma\mathbf 1_L
=
\mathbf 1_{N_\sigma},
\]
which proves that $\mathbf H_\sigma$ is row stochastic. 
Finally, comparing
\vspace{-0.25cm}
\[
\mathbf r_{\mathcal W_\sigma}
=
\left(\mathbf H_\sigma\otimes \mathbf I_3\right)
\mathbf r_{\mathcal A}
\]
with
\vspace{-0.25cm}
\[
\mathbf p_{b_i,\sigma}
=
\sum_{a\in\mathcal A}
\alpha_{b_i a}^{\sigma}
\mathbf a_{a,\sigma},
\]
shows that
\vspace{-0.25cm}
\[
[\mathbf H_\sigma]_{ij}
=
\alpha_{b_i a_j}^{\sigma},
\]
i.e., the $i$th row of $\mathbf H_\sigma$ contains the final
anchor-coordinate weights of worker agent $b_i$. 
\end{proof}
\section{Network Dynamics and Control}\label{Network Dynamics and Control}

Worker agent $i\in\mathcal W_\sigma$ is considered to be a multi-copter UAS and  modeled by
\vspace{-0.25cm}
\begin{equation}\label{orginaldyn}
\begin{cases}
\dot{\mathbf r}_i=\mathbf v_i,\\
\dot{\mathbf v}_i
=
-g\hat{\mathbf e}_3
+\dfrac{mg+T_i}{m}
\mathbf R(\boldsymbol{\eta}_i)\hat{\mathbf e}_3,\\
\dot{\boldsymbol{\eta}}_i
=
\mathbf E(\phi_i,\theta_i)\boldsymbol{\omega}_i,\\
\mathbf I\dot{\boldsymbol{\omega}}_i
+
\boldsymbol{\omega}_i\times(\mathbf I\boldsymbol{\omega}_i)
=
\boldsymbol{\tau}_i,\\
\dot T_i=\dot T_i,\\
\ddot T_i=u_{T,i},
\end{cases}
\end{equation}
where
$\mathbf r_i,\mathbf v_i\in\mathbb R^3$
denote the inertial position and velocity,
$\boldsymbol{\eta}_i=(\phi_i,\theta_i,\psi_i)$
are the Euler angles,
$\boldsymbol{\omega}_i=(p_i,q_i,r_i)$
is the body angular velocity,
$T_i$ denotes the thrust deviation from hover,
$\boldsymbol{\tau}_i=(\tau_{\phi,i},\tau_{\theta,i},\tau_{\psi,i})^\top$
is the control torque,
$\hat{\mathbf e}_3=[0~0~1]^\top$,
and
$\mathbf E(\phi_i,\theta_i)$
is the Euler kinematic matrix.
\subsection{Nonlinear Control-Affine Dynamics}
Differentiating
\vspace{-0.25cm}
\begin{equation}
\dot{\boldsymbol{\eta}}_i
=
\mathbf E(\phi_i,\theta_i)\boldsymbol{\omega}_i
\end{equation}
gives
\vspace{-0.25cm}
\begin{equation}
\ddot{\boldsymbol{\eta}}_i
=
\mathbf E(\phi_i,\theta_i)\dot{\boldsymbol{\omega}}_i
+
\dot{\mathbf E}(\phi_i,\theta_i)\boldsymbol{\omega}_i.
\end{equation}
Define the virtual attitude input
\vspace{-0.25cm}
\begin{equation}
\mathbf u_{\eta,i}
=
\begin{bmatrix}
u_{\phi,i}\\
u_{\theta,i}\\
u_{\psi,i}
\end{bmatrix}
=
\ddot{\boldsymbol{\eta}}_i.
\end{equation}
Solving for
$\dot{\boldsymbol{\omega}}_i$
yields
\vspace{-0.25cm}
\begin{equation}
\dot{\boldsymbol{\omega}}_i
=
\mathbf E^{-1}(\phi_i,\theta_i)
\left(
\mathbf u_{\eta,i}
-
\dot{\mathbf E}(\phi_i,\theta_i)\boldsymbol{\omega}_i
\right),
\end{equation}
and substituting into the rigid-body dynamics gives
\vspace{-0.25cm}
\begin{equation}
\boldsymbol{\tau}_i
=
\mathbf I
\mathbf E^{-1}(\phi_i,\theta_i)\mathbf u_{\eta,i}
-
\mathbf I
\mathbf E^{-1}(\phi_i,\theta_i)
\dot{\mathbf E}(\phi_i,\theta_i)\boldsymbol{\omega}_i
+
\boldsymbol{\omega}_i\times(\mathbf I\boldsymbol{\omega}_i).
\end{equation}
Define the state and input vectors
\vspace{-0.25cm}
\begin{equation}
\mathbf x_i
=
\begin{bmatrix}
\mathbf r_i^\top&
\mathbf v_i^\top&
\boldsymbol{\eta}_i^\top&
\dot{\boldsymbol{\eta}}_i^\top&
T_i&
\dot T_i
\end{bmatrix}^{\!\top}
\in\mathbb R^{14},
\end{equation}
\begin{equation}
\mathbf u_i
=
\begin{bmatrix}
u_{T,i}&
\mathbf u_{\eta,i}^{\top}
\end{bmatrix}^{\!\top}
\in\mathbb R^4.
\end{equation}

Using
$\boldsymbol{\omega}_i
=
\mathbf E^{-1}(\phi_i,\theta_i)\dot{\boldsymbol{\eta}}_i$,
the dynamics admit the control-affine form
\vspace{-0.25cm}
\begin{equation}
\resizebox{0.99\hsize}{!}{%
$\dot{\mathbf x}_i
=
\begin{bmatrix}
\mathbf v_i\\
-\,g\hat{\mathbf e}_3
+\dfrac{mg+T_i}{m}
\mathbf R(\boldsymbol{\eta}_i)\hat{\mathbf e}_3\\
\dot{\boldsymbol{\eta}}_i\\
\mathbf 0_3\\
\dot T_i\\
0
\end{bmatrix}
+
\begin{bmatrix}
\mathbf v_i\\
-\,g\hat{\mathbf e}_3
+\dfrac{mg+T_i}{m}
\mathbf R(\boldsymbol{\eta}_i)\hat{\mathbf e}_3\\
\dot{\boldsymbol{\eta}}_i\\
\mathbf 0_3\\
\dot T_i\\
0
\end{bmatrix}\mathbf u_i.
$
}
\end{equation}
The multi-copter dynamics in \eqref{orginaldyn} are differentially flat with flat outputs given by the inertial position
$\mathbf r_i=[x_i~y_i~z_i]^\top$
and yaw angle
$\psi_i$.
Consequently, all system states and control inputs, including the roll and pitch angles, body angular velocity, thrust, and control torques, can be expressed as algebraic functions of the flat outputs and a finite number of their time derivatives. This property enables the decomposition of the vehicle dynamics into external dynamics and internal attitude dynamics. The external dynamics govern the evolution of the flat outputs and are used for network-level coordination and trajectory generation, while the internal dynamics are responsible for realizing the corresponding thrust and attitude commands.

Accordingly, define the external state as
\vspace{-0.25cm}
\begin{equation}
\mathbf z_i
=
\begin{bmatrix}
\mathbf r_i^\top &
\dot{\mathbf r}_i^\top &
\ddot{\mathbf r}_i^\top &
\dddot{\mathbf r}_i^\top &
\psi_i &
\dot{\psi}_i
\end{bmatrix}^{\top},
\end{equation}
and the external input as
\vspace{-0.25cm}
\begin{equation}
\mathbf s_i
=
\begin{bmatrix}
\ddddot{\mathbf r}_i\\
\ddot{\psi}_i
\end{bmatrix}.
\end{equation}
The external dynamics are therefore described by the linear state-space model
\vspace{-0.25cm}
\begin{equation}
\dot{\mathbf z}_i
=
\mathbf A_{\rm EXT}\mathbf z_i
+
\mathbf B_{\rm EXT}\mathbf s_i,
\end{equation}
where the matrices
$\mathbf A_{\rm EXT}$
and
$\mathbf B_{\rm EXT}$
are identical to those derived in \cite{el2023quadcopter,rastgoftar2026deep}.

The translational subsystem forms a fourth-order chain of integrators
driven by snap, while the yaw subsystem is second order. The external
feedback law is selected as
\vspace{-0.25cm}
\begin{equation}
\resizebox{0.99\hsize}{!}{%
$
\begin{aligned}
\mathbf s_i
=&
-k_{13,i}\dot\psi_i
-k_{14,i}\psi_i
+\mathbf K_{jerk,i}
(\dddot{\mathbf r}_{i,d}-\mathbf j_i)\\
&
+\mathbf K_{acc,i}
(\ddot{\mathbf r}_{i,d}-\mathbf a_i)
+\mathbf K_{vel,i}
(\dot{\mathbf r}_{i,d}-\mathbf v_i)
+\mathbf K_{pos,i}
(\mathbf r_{i,d}-\mathbf r_i),
\end{aligned}
$
}
\end{equation}
where
$\mathbf K_{jerk,i}$,
$\mathbf K_{acc,i}$,
$\mathbf K_{vel,i}$,
and
$\mathbf K_{pos,i}$
are positive-definite diagonal gain matrices. The gain bounds are selected such that the resulting closed-loop
external dynamics are asymptotically stable. Differentiating the translational dynamics twice yields the
affine relation
\vspace{-0.25cm}
\begin{equation}
\mathbf s_i
=
\mathbf M(\mathbf x_i)\mathbf u_i
+
\mathbf n(\mathbf x_i),
\end{equation}
where
\vspace{-0.25cm}
\begin{equation}
\mathbf M(\mathbf x_i)
=
\begin{bmatrix}
\mathbf M'(\mathbf x_i)\\
0\;\;0\;\;0\;\;1
\end{bmatrix},
\qquad
\mathbf n(\mathbf x_i)
=
\begin{bmatrix}
\mathbf n'(\mathbf x_i)\\
0
\end{bmatrix},
\end{equation}
with
$\mathbf M'(\mathbf x_i)$
and
$\mathbf n'(\mathbf x_i)$
obtained from the expressions in \cite{el2023quadcopter} by replacing the
state variables with those of agent $i$.
\subsection{Network Dynamics}
By setting the desired yaw angle to zero, the yaw gains
$k_{13,i}$ and $k_{14,i}$ are selected independently for every
worker agent
$i\in\mathcal W_\sigma$
such that
$
\psi_i(t)\rightarrow 0$. Consequently, the translational and yaw dynamics become decoupled, and the
network dynamics can be expressed solely in terms of the translational
positions.

Define the stacked worker-position vector
\vspace{-0.25cm}
\begin{equation}
\mathbf y
=
\begin{bmatrix}
\mathbf r_{b_1}^\top&
\mathbf r_{b_2}^\top&
\cdots&
\mathbf r_{b_{N_\sigma}}^\top
\end{bmatrix}^{\!\top}
\in\mathbb R^{2N_\sigma},
\end{equation}
and the stacked anchor-position vector
\vspace{-0.25cm}
\begin{equation}
\mathbf y_L
=
\left(
\mathbf B_\sigma
\otimes
\mathbf I_2
\right)
\begin{bmatrix}
\mathbf r_{a_1}^\top&
\cdots&
\mathbf r_{a_L}^\top
\end{bmatrix}^{\!\top}
\in\mathbb R^{2N_\sigma},
\end{equation}
where $\otimes$ denotes the Kronecker product. Furthermore, define
\vspace{-0.25cm}
\begin{equation}
\mathbf A
=
\mathbf A_\sigma
\otimes
\mathbf I_3,
\end{equation}
and let $\mathbf K_{\rm pos}$, $\mathbf K_{\rm vel}$, $\mathbf K_{\rm acc}$, and $\mathbf K_{\rm jerk}$
denote the corresponding block-diagonal gain matrices. The closed-loop
network dynamics are then given by
\vspace{-0.25cm}
\begin{equation}
\ddddot{\mathbf y}
=
-
\mathbf K_{\rm jerk}\dddot{\mathbf y}
-
\mathbf K_{\rm acc}\ddot{\mathbf y}
-
\mathbf K_{\rm vel}\dot{\mathbf y}
+
\mathbf K_{\rm pos}
\left(
\mathbf A\mathbf y
+
\mathbf y_L
\right).
\end{equation}

\begin{theorem}\label{thm:network_convergence}
Consider mode $\sigma\in\mathcal R$ and suppose that the communication graph
is generated by the DNN construction, so that \eqref{picom} and \eqref{strictlylayered} hold. Let the closed-loop translational network
dynamics be
\vspace{-0.25cm}
\begin{equation}
\ddddot{\mathbf y}
=
-\mathbf K_{\rm jerk}\dddot{\mathbf y}
-\mathbf K_{\rm acc}\ddot{\mathbf y}
-\mathbf K_{\rm vel}\dot{\mathbf y}
+
\mathbf K_{\rm pos}
\left(
\mathbf A_\sigma(k)\mathbf y+\mathbf y_L
\right),
\end{equation}
where $\mathbf y$ is the stacked worker-position vector and $\mathbf y_L$
is the anchor-induced input. Assume that after
$k=k_\sigma^{\rm mid}$, the communication weights are constant and satisfy
$w_{ij}^{\sigma}=1/3$. Let
\vspace{-0.25cm}
\begin{equation}
\mathbf A_\sigma
=
\mathbf A_\sigma(k_\sigma^{\rm mid}).
\end{equation}
Assume further that the gain matrices are block diagonal and positive
definite:
\vspace{-0.25cm}
\begin{equation}
\mathbf K_{\rm pos}>0,\quad
\mathbf K_{\rm vel}>0,\quad
\mathbf K_{\rm acc}>0,\quad
\mathbf K_{\rm jerk}>0.
\end{equation}
If, for every eigenvalue $\mu_q$ of $\mathbf A_\sigma$, the polynomial
\vspace{-0.25cm}
\begin{equation}
s^4
+
k_{{\rm jerk},i}s^3
+
k_{{\rm acc},i}s^2
+
k_{{\rm vel},i}s
-
k_{{\rm pos},i}\mu_q
\end{equation}
is Hurwitz for every worker agent $i\in\mathcal W_\sigma$, then the
network dynamics are exponentially stable with respect to the equilibrium
\vspace{-0.25cm}
\begin{equation}
\mathbf y^\star
=
-\mathbf A_\sigma^{-1}\mathbf y_L.
\end{equation}
Consequently,
\[
\lim_{t\to\infty}\mathbf y(t)=\mathbf y^\star,~
\lim_{t\to\infty}\dot{\mathbf y}(t)=\mathbf 0,~
\lim_{t\to\infty}\ddot{\mathbf y}(t)=\mathbf 0,~
\lim_{t\to\infty}\dddot{\mathbf y}(t)=\mathbf 0.
\]
Moreover, the equilibrium satisfies
\begin{equation}
\mathbf y^\star
=
\mathbf H_\sigma\mathbf y_{\mathcal A},
\qquad
\mathbf H_\sigma=-\mathbf A_\sigma^{-1}\mathbf B_\sigma,
\end{equation}
and therefore every worker converges to the anchor-coordinate position
prescribed by the DNN communication topology.
\end{theorem}

\begin{proof}
Since the communication graph is generated by the DNN layering procedure,
the worker-worker communication block $\mathbf A_\sigma$ is permutation
similar to a lower triangular matrix with diagonal entries equal to $-1$.
Hence, $\mathbf A_\sigma$ is nonsingular and Hurwitz.

At $k\ge k_\sigma^{\rm mid}$, the communication weights are constant.
Therefore, the anchor-induced term $\mathbf y_L$ is constant because the
anchor agents are stationary. Define the equilibrium
\vspace{-0.25cm}
\begin{equation}
\mathbf y^\star
=
-\mathbf A_\sigma^{-1}\mathbf y_L.
\end{equation}
Then
\vspace{-0.25cm}
\begin{equation}
\mathbf A_\sigma\mathbf y^\star+\mathbf y_L=\mathbf 0.
\end{equation}
Let
\vspace{-0.25cm}
\begin{equation}
\tilde{\mathbf y}
=
\mathbf y-\mathbf y^\star.
\end{equation}
Substituting into the network dynamics gives the homogeneous error system
\vspace{-0.25cm}
\begin{equation}
\ddddot{\tilde{\mathbf y}}
=
-\mathbf K_{\rm jerk}\dddot{\tilde{\mathbf y}}
-\mathbf K_{\rm acc}\ddot{\tilde{\mathbf y}}
-\mathbf K_{\rm vel}\dot{\tilde{\mathbf y}}
+
\mathbf K_{\rm pos}\mathbf A_\sigma\tilde{\mathbf y}.
\end{equation}

Equivalently, defining
\vspace{-0.25cm}
\begin{equation}
\mathbf z
=
\begin{bmatrix}
\tilde{\mathbf y}^{\top}&
\dot{\tilde{\mathbf y}}^{\top}&
\ddot{\tilde{\mathbf y}}^{\top}&
\dddot{\tilde{\mathbf y}}^{\top}
\end{bmatrix}^{\top},
\end{equation}
the error dynamics can be written as
\vspace{-0.25cm}
\begin{equation}
\dot{\mathbf z}
=
\mathbf F_\sigma\mathbf z,
\end{equation}
where
\vspace{-0.25cm}
\begin{equation}
\mathbf F_\sigma
=
\begin{bmatrix}
\mathbf 0 & \mathbf I & \mathbf 0 & \mathbf 0\\
\mathbf 0 & \mathbf 0 & \mathbf I & \mathbf 0\\
\mathbf 0 & \mathbf 0 & \mathbf 0 & \mathbf I\\
\mathbf K_{\rm pos}\mathbf A_\sigma
&
-\mathbf K_{\rm vel}
&
-\mathbf K_{\rm acc}
&
-\mathbf K_{\rm jerk}
\end{bmatrix}.
\end{equation}
By assumption, all modal characteristic polynomials associated with
$\mathbf F_\sigma$ are Hurwitz. Therefore, $\mathbf F_\sigma$ is Hurwitz,
and the origin of the error system is exponentially stable. Hence,
\vspace{-0.25cm}
\begin{equation}
\lim_{t\to\infty}\tilde{\mathbf y}(t)=\mathbf 0,
~
\lim_{t\to\infty}\dot{\tilde{\mathbf y}}(t)=\mathbf 0,
~
\lim_{t\to\infty}\ddot{\tilde{\mathbf y}}(t)=\mathbf 0,
~
\lim_{t\to\infty}\dddot{\tilde{\mathbf y}}(t)=\mathbf 0.
\end{equation}
Thus,
\vspace{-0.25cm}
\begin{equation}
\lim_{t\to\infty}\mathbf y(t)=\mathbf y^\star.
\end{equation}

Finally, from the steady-state communication relation
\begin{equation}
\mathbf B_\sigma\mathbf y_{\mathcal A}
+
\mathbf A_\sigma\mathbf y^\star
=
\mathbf 0,
\end{equation}
we obtain
\begin{equation}
\mathbf y^\star
=
-\mathbf A_\sigma^{-1}\mathbf B_\sigma\mathbf y_{\mathcal A}
=
\mathbf H_\sigma\mathbf y_{\mathcal A}.
\end{equation}
This proves convergence to the DNN-induced anchor-coordinate configuration.
\end{proof}
\section{Information-Theoretic Coverage Augmentation}
\label{Information-Theoretic Coverage Augmentation}

Assume that the anchor-agent set $\mathcal A$, the worker-agent sets
$\{\mathcal W_\sigma\}_{\sigma\in\mathcal R}$, and the mode-transition graph
$\mathcal{G}_{\mathrm{mode}}$ are given. For each mode transition
$\sigma\rightarrow\sigma^+$, the final desired reference positions of the
persistent worker agents, defined by $
i\in
\mathcal W_\sigma
\cap
\mathcal W_{\sigma^+}
$
are available, whereas the reference positions of the newly introduced worker
agents, defined by $\mathcal W_{\sigma^+}
\setminus
\mathcal W_\sigma$
must be determined. Let $
P
=
|\mathcal W_{\sigma^+}
\setminus
\mathcal W_\sigma|
$
denote the number of newly introduced worker agents. The objective is to
select their reference positions so that the spatial distribution of the
augmented worker-agent set closely matches the spatial distribution of the
surveillance nodes.

Recall that the surveillance region is discretized into the surveillance-node 
set $\mathcal{D}$,
where $\mathbf d_i\in\mathbb R^2$ denotes the location of the $i\in \mathcal{D}$th surveillance node.
The desired surveillance-node distribution is modeled by the discrete
probability distribution
\begin{equation}
\rho_D(i)
=
\frac{1}{M},
\qquad
i=1,\ldots,M,
\end{equation}
which assigns equal probability to every surveillance node.

The spatial distribution induced by the
reference positions of the worker agents is approximated by the Gaussian kernel
density estimator
\vspace{-0.25cm}
\begin{equation}
\rho_A(i)
=
\frac{
\displaystyle
\sum_{j\in\mathcal W_{\sigma^+}}
K_h
\!\left(
d_i-\mathbf a_j
\right)
}{
\displaystyle
\sum_{\ell=1}^{M}
\sum_{j\in\mathcal W_{\sigma^+}}
K_h
\!\left(
d_\ell-\mathbf a_j
\right)
},
\end{equation}
where
\vspace{-0.25cm}
\begin{equation}
K_h(\mathbf x)
=
\exp
\left(
-\frac{\|\mathbf x\|^2}{2h^2}
\right)
\end{equation}
is a Gaussian kernel with bandwidth $h>0$.

The mismatch between the surveillance-node distribution and the induced
worker-agent distribution is quantified by the Kullback--Leibler divergence
\vspace{-0.25cm}
\begin{equation}
D_{\rm KL}
(\rho_D\|\rho_A)
=
\sum_{i=1}^{M}
\rho_D(i)
\log
\frac{
\rho_D(i)
}{
\rho_A(i)+\epsilon
},
\end{equation}
where $\epsilon>0$ is a small regularization constant introduced to avoid
numerical singularities.

Rather than solving a continuous nonlinear optimization problem, the proposed
algorithm constructs the augmented worker-agent configuration sequentially.
During each iteration, every surveillance node is treated as a candidate
reference position for the next newly introduced worker agent. For each
candidate location, the worker-agent density is recomputed after temporarily
adding the candidate to the current augmented worker-agent set, and the
corresponding Kullback--Leibler divergence is evaluated. The candidate yielding
the smallest divergence is selected as the next reference position. The
selected worker agent is then permanently added to the augmented worker-agent
set, and the procedure is repeated until all $P$ newly introduced worker agents
have been placed.

Algorithmically, the reference position of the $p$th newly introduced worker
agent is determined by
\vspace{-0.25cm}
\begin{equation}
\mathbf a_p
=
\arg\min_{d_i\in\mathcal D}
D_{\rm KL}
\!\left(
\rho_D
\big\|
\rho_A^{(p,i)}
\right),
\end{equation}
where $\rho_A^{(p,i)}$ denotes the worker-agent density obtained by
temporarily assigning the $p$th worker agent to the candidate surveillance node
$d_i$ while keeping the remaining worker-agent positions unchanged.

\begin{table}[t]
\centering
\caption{Active and replaced agents for each operation mode.}
\label{tab:rotation_modes}
\begin{tabular}{|c|c|c|}
\hline
Mode $\sigma$ &
Active agents $\mathcal{V}_\sigma$ &
Replaced agents $\mathcal{V}\setminus\mathcal{V}_\sigma$ \\
\hline
1,~5&
$\{1,\ldots,14\}$ &
$\{15,16,17\}$ \\
\hline
2,~6 &
$\{1,\ldots,11,15,16,17\}$ &
$\{12,13,14\}$ \\
\hline
3,~7 &
$\{1,\ldots,8,12,\ldots,17\}$ &
$\{9,10,11\}$ \\
\hline
4,~8 &
$\{1,\ldots,5,9,\ldots,17\}$ &
$\{6,7,8\}$ \\
\hline
\end{tabular}
\end{table}

\begin{figure}[ht]
\center
\includegraphics[width=3.3 in]{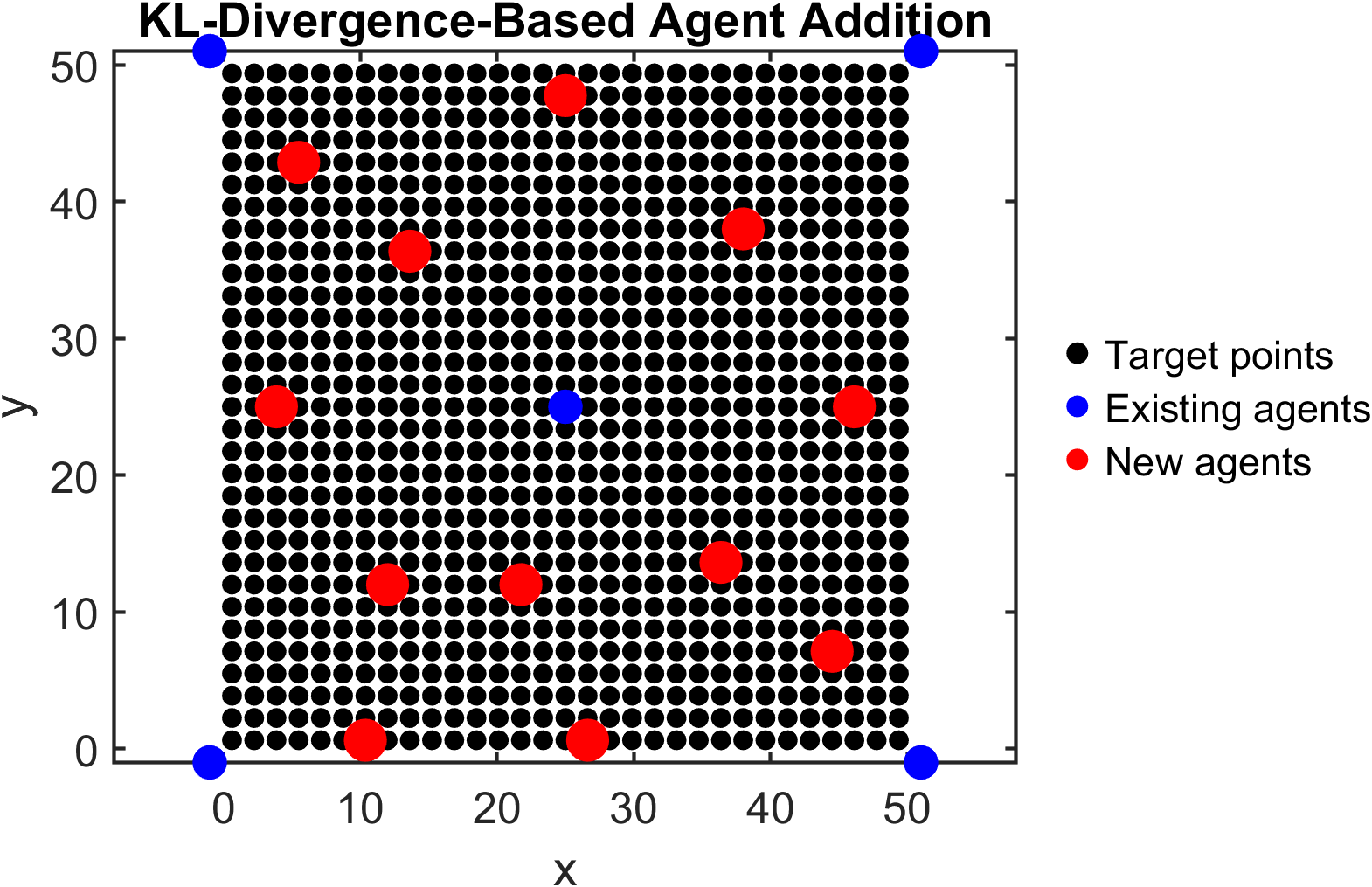}
\vspace{-.4cm}
\caption{Initial reference configuration for the surveillance mission.
Black dots denote the discretized surveillance nodes comprising the initial
unsearched set
\(\mathcal U_1=\mathcal D\),
while red markers indicate the reference positions of the worker UAS obtained
using the proposed KL divergence-based initialization method. The anchor UAS
remain fixed throughout the mission and define the reference geometry for
distributed coordination.}
\label{KL_initial}
\end{figure} 
\begin{figure*}[h]
\centering
 \subfigure[]{\includegraphics[width=0.24\linewidth]{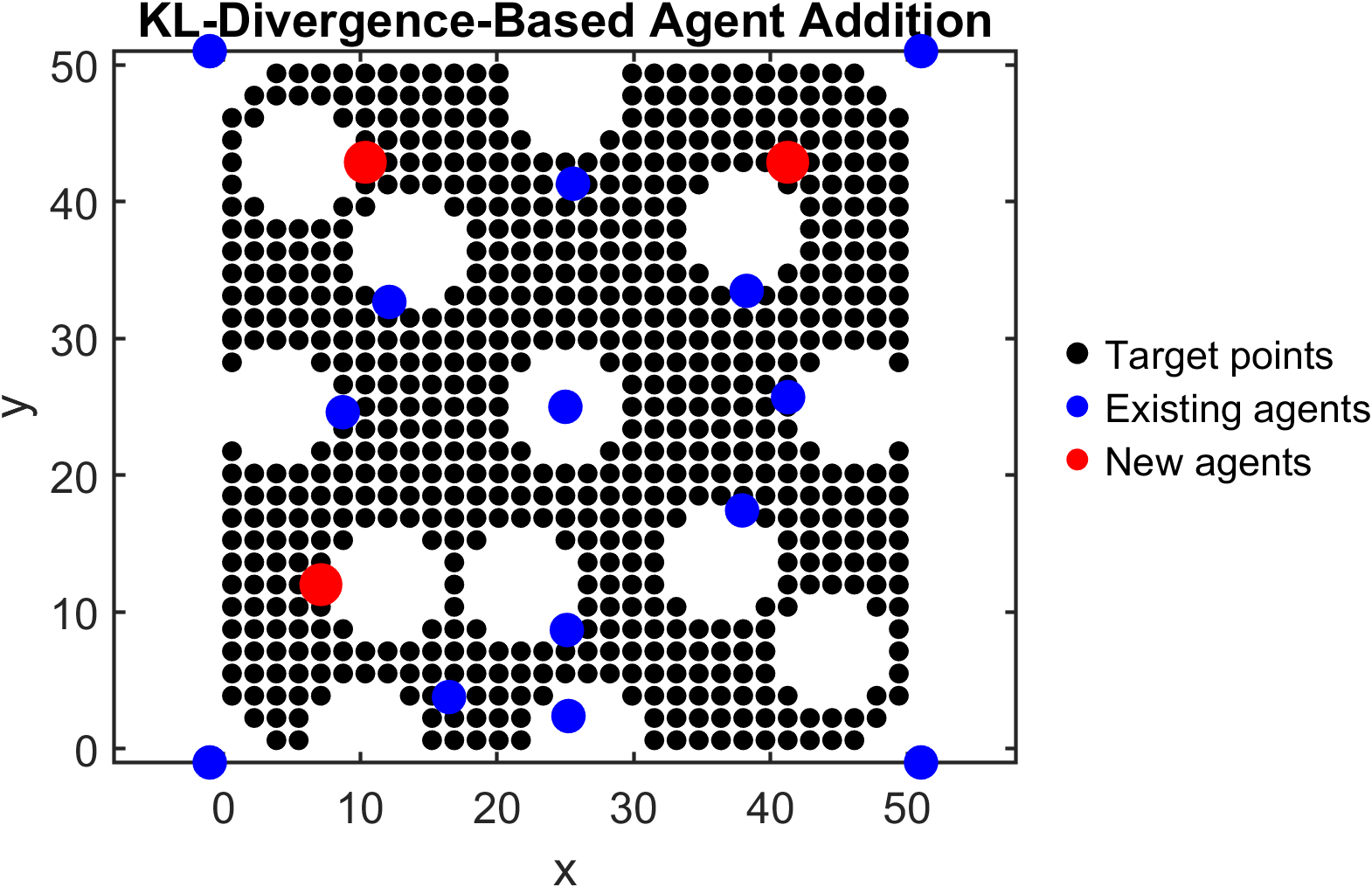}}
 \subfigure[]{\includegraphics[width=0.24\linewidth]{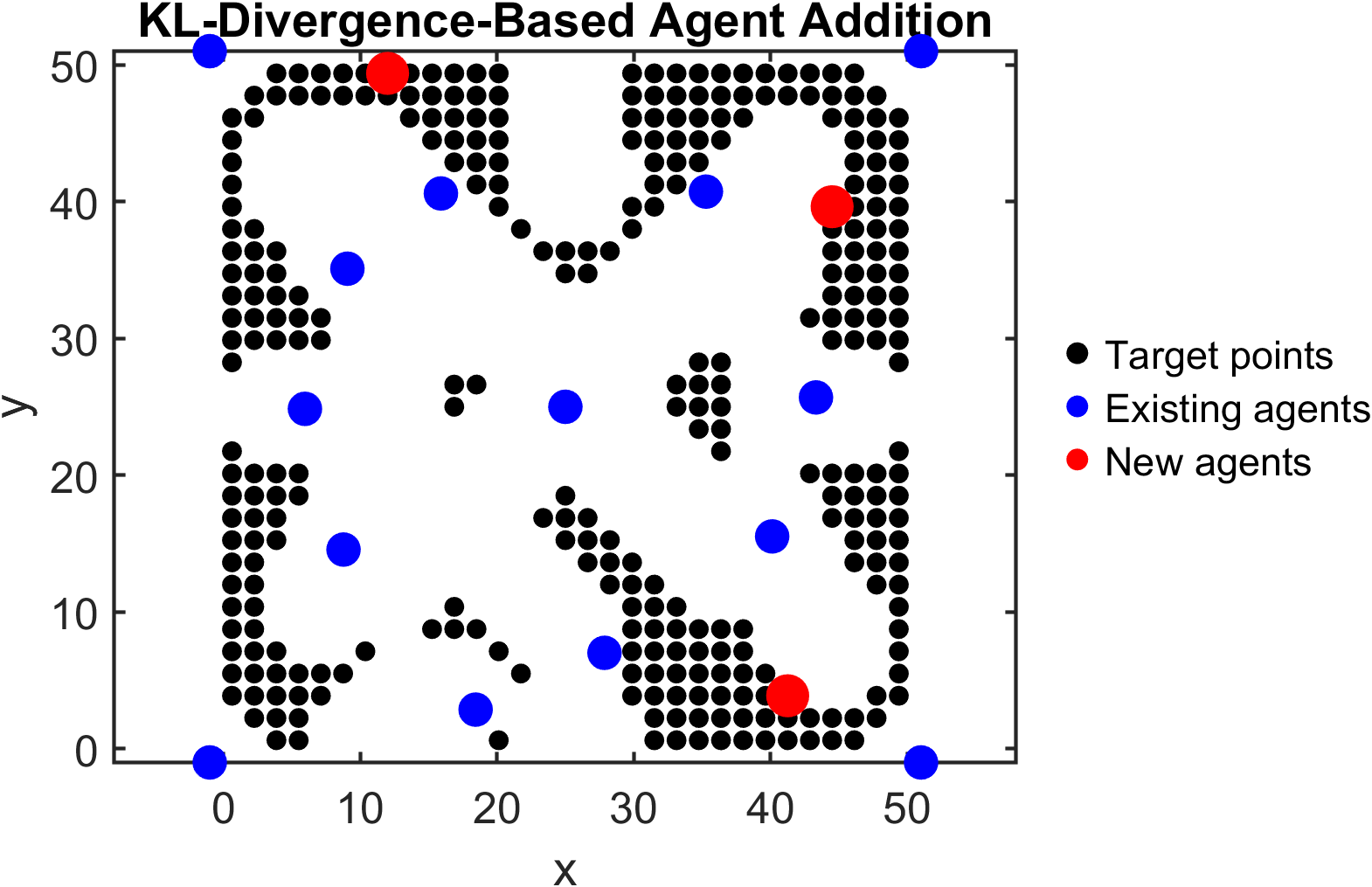}}
 \subfigure[]{\includegraphics[width=0.24\linewidth]{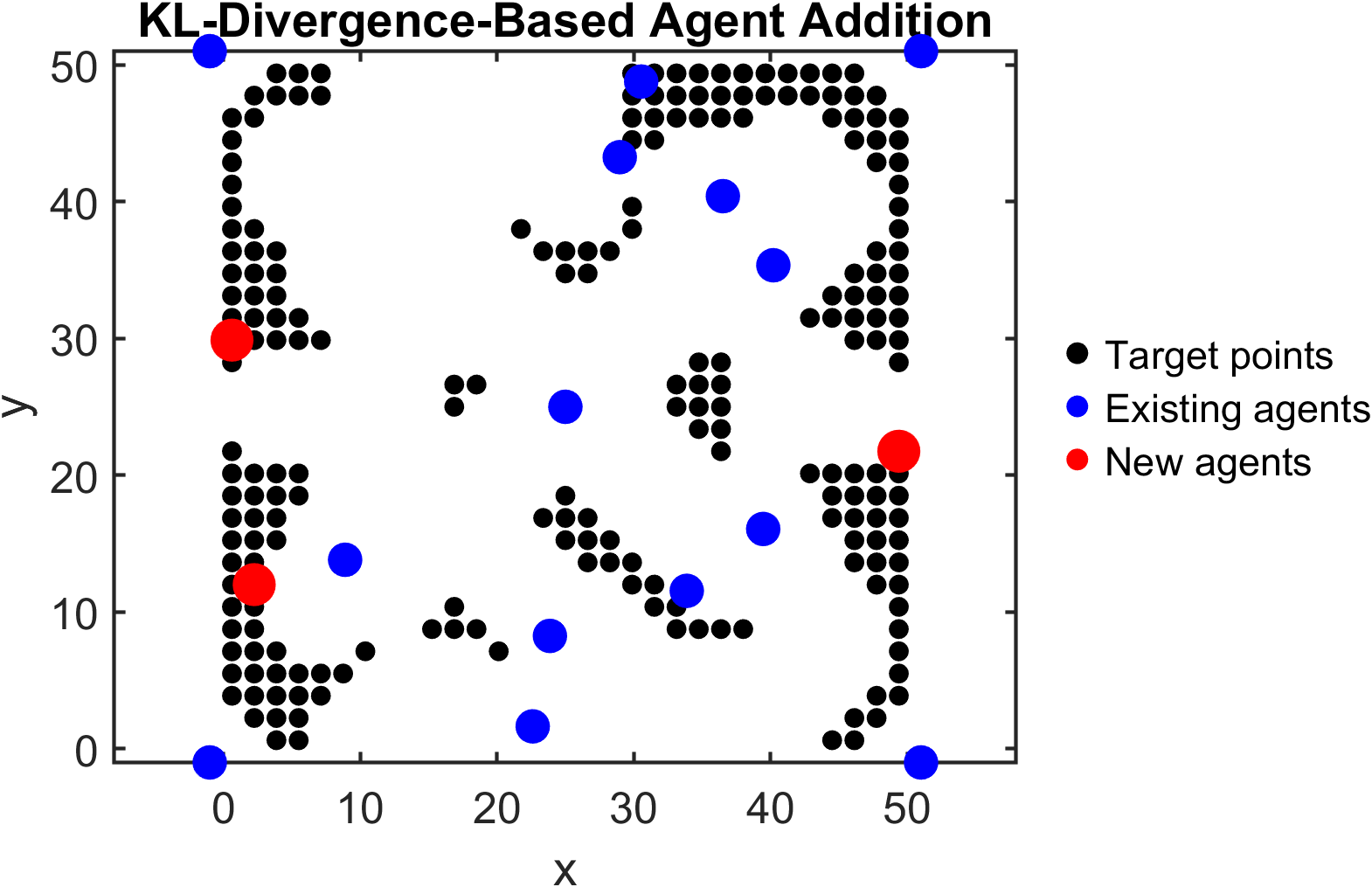}}
 \subfigure[]{\includegraphics[width=0.24\linewidth]{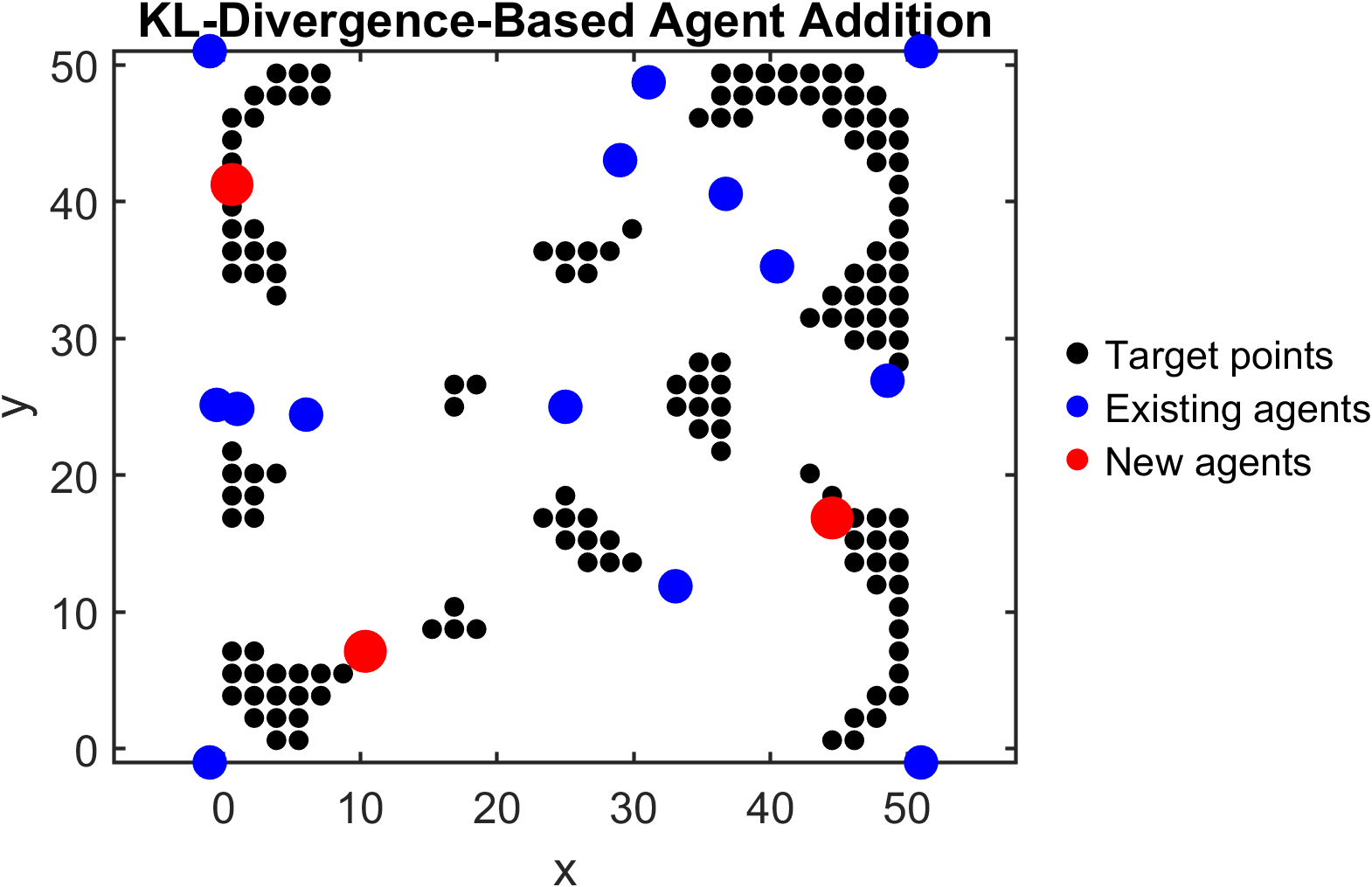}}
 \subfigure[]{\includegraphics[width=0.24\linewidth]{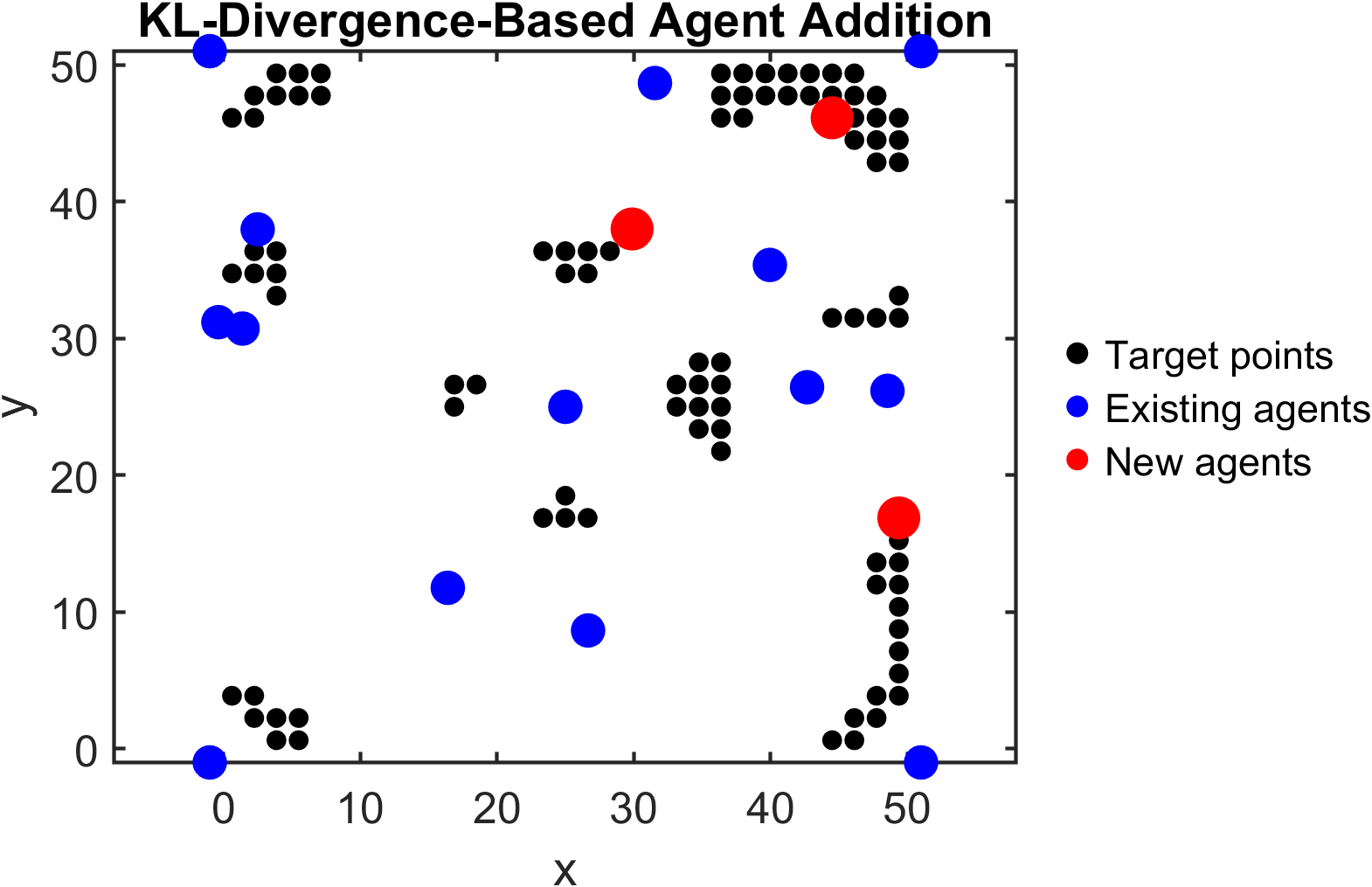}}
 \subfigure[]{\includegraphics[width=0.24\linewidth]{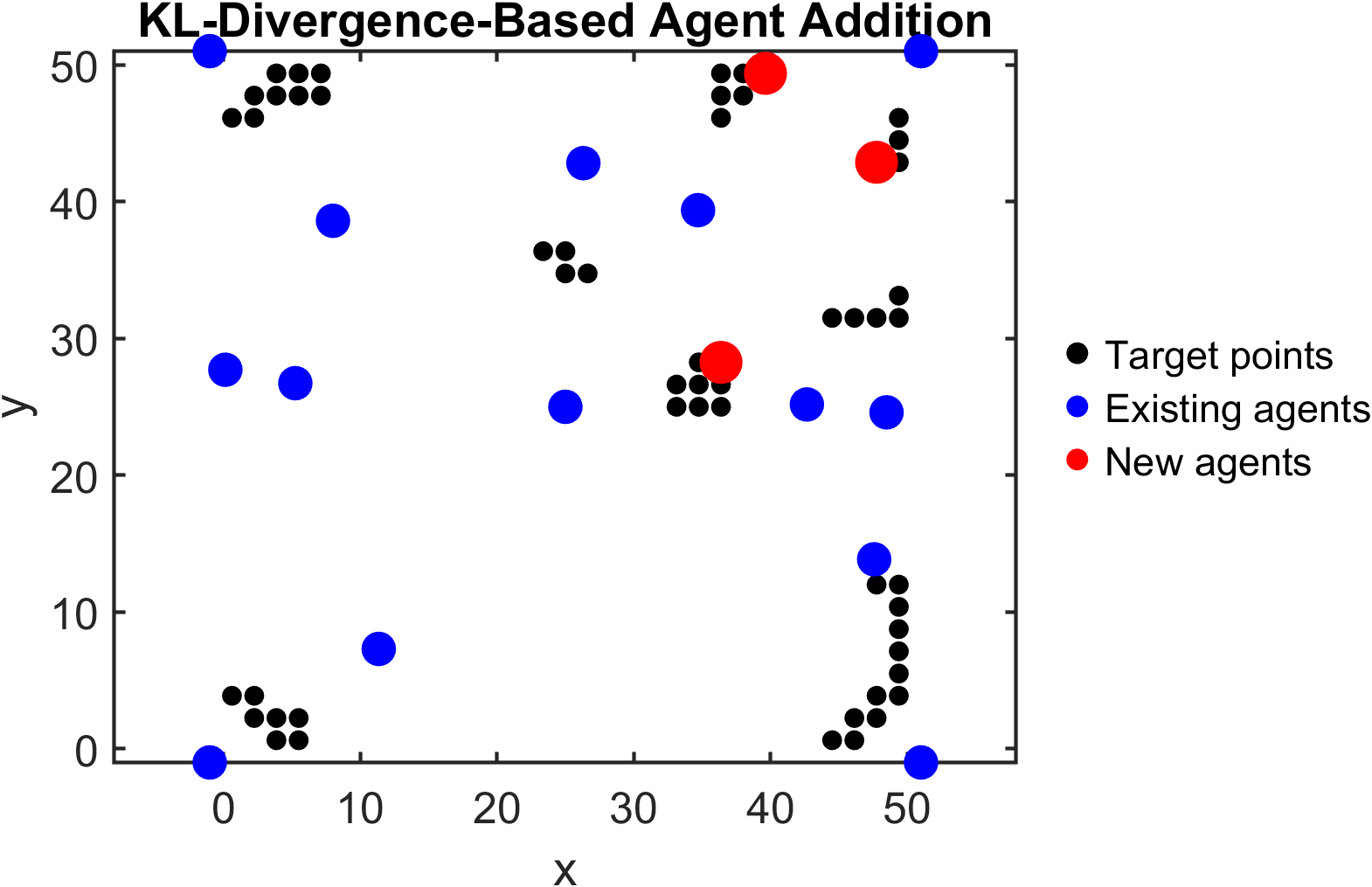}}
 \subfigure[]{\includegraphics[width=0.24\linewidth]{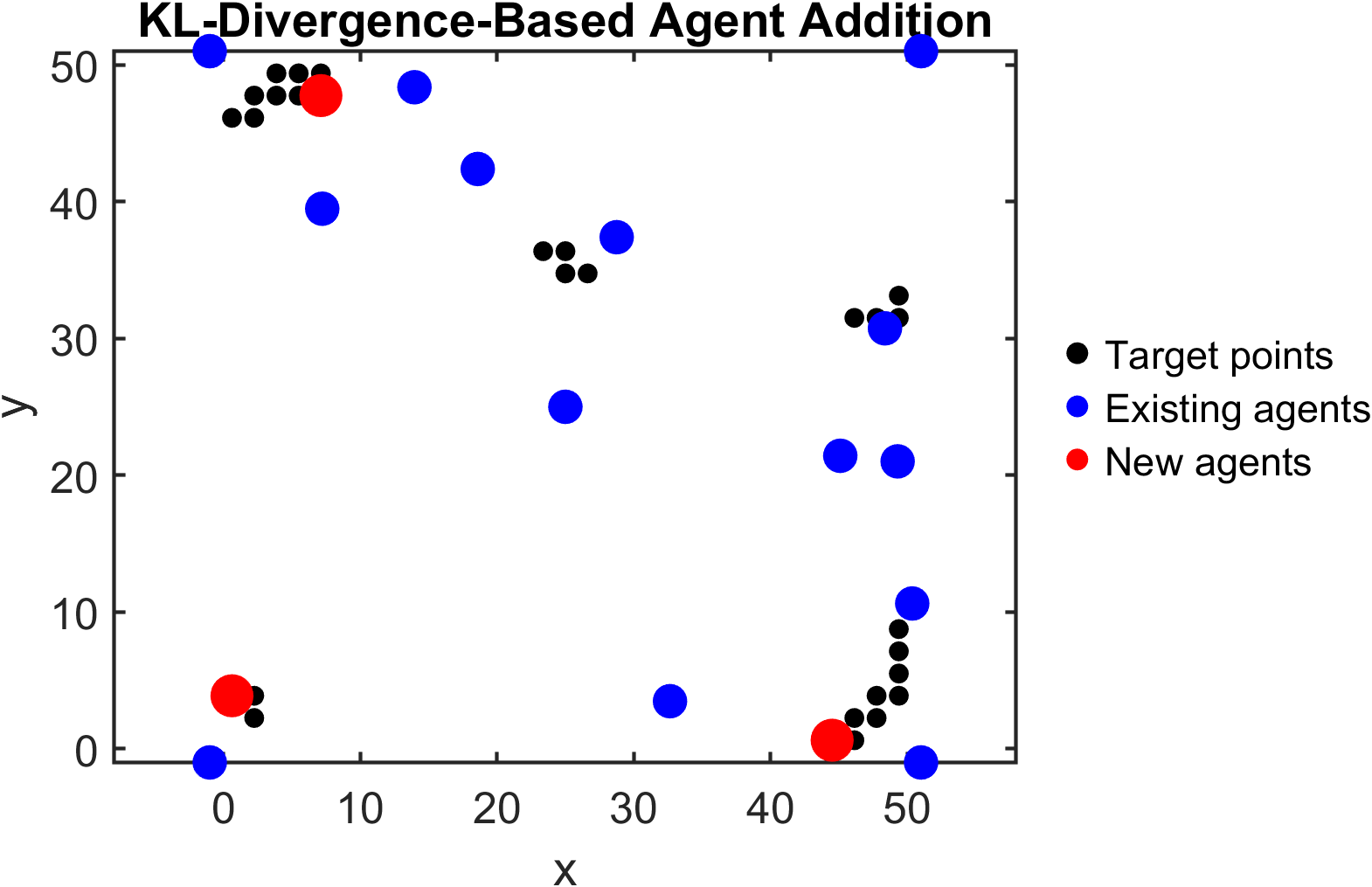}}
 \subfigure[]{\includegraphics[width=0.24\linewidth]{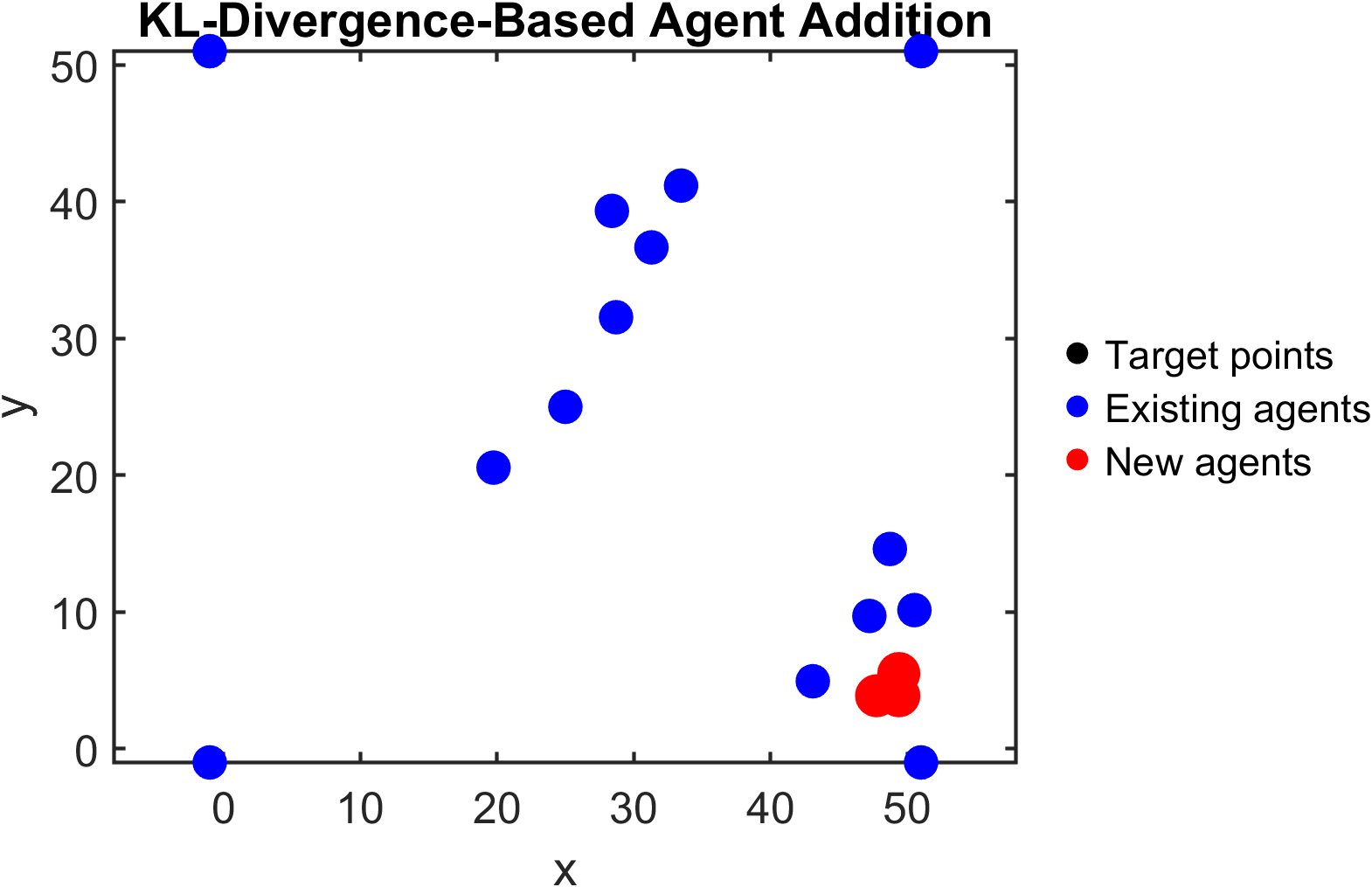}} 
 \vspace{-0.4cm}
\caption{Reference configurations generated for the successive operation
modes. For each mode
\(\sigma\),
the figure illustrates the reference configuration of the subsequent mode
\(\sigma^{+}\).
Black markers denote the unsearched target nodes remaining after completion
of mode
\(\sigma\),
while red markers indicate the reference positions assigned to the replacement
worker UAS. The resulting configurations maximize the enclosed probability
mass of the remaining target distributions while preserving the required
communication topology and geometric constraints for decentralized
coordination.}
\label{referenceall}
\end{figure*}
\begin{figure}[h]
\centering
 \subfigure[]{\includegraphics[width=0.49\linewidth]{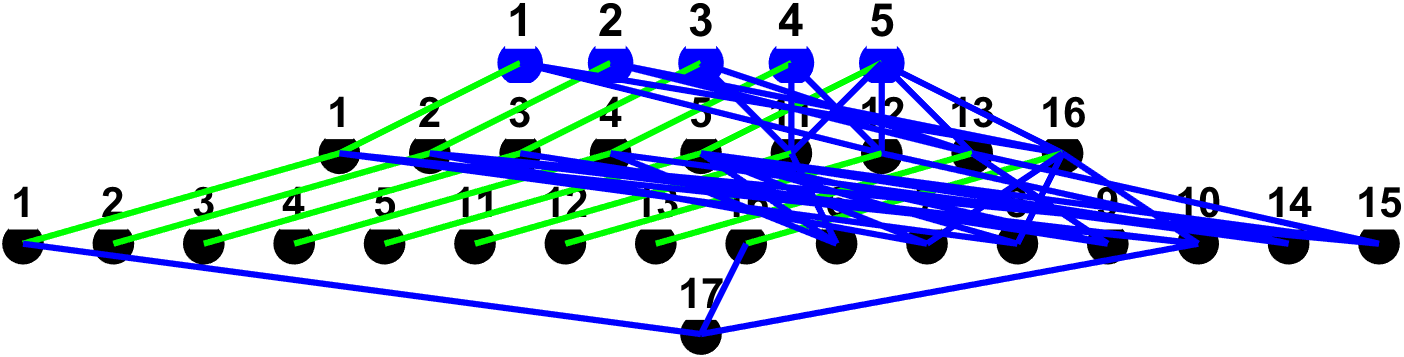}}
 \subfigure[]{\includegraphics[width=0.49\linewidth]{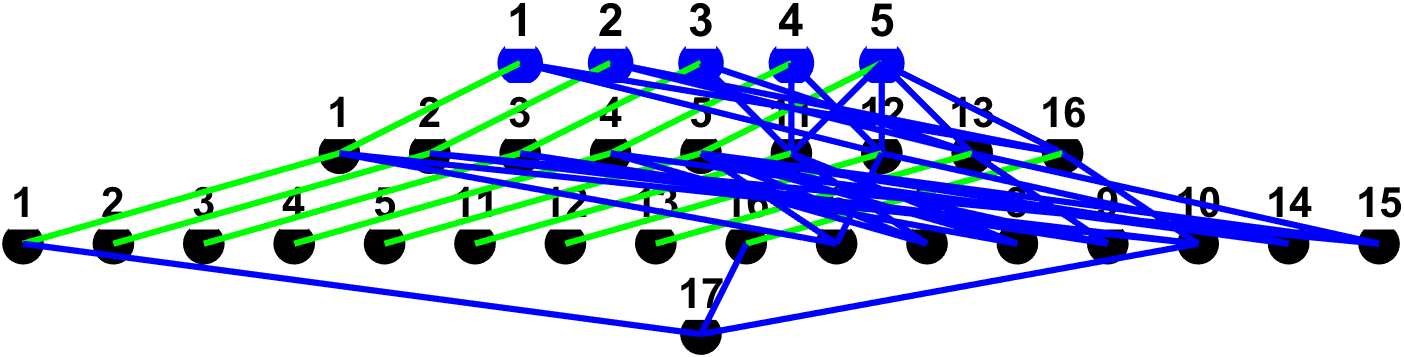}}
 \subfigure[]{\includegraphics[width=0.49\linewidth]{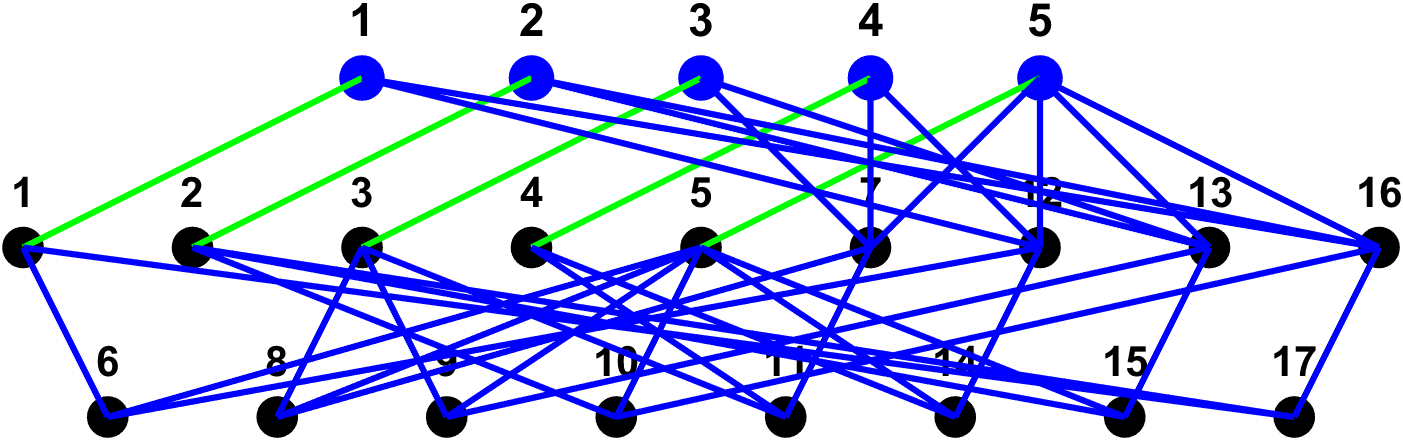}}
 \subfigure[]{\includegraphics[width=0.49\linewidth]{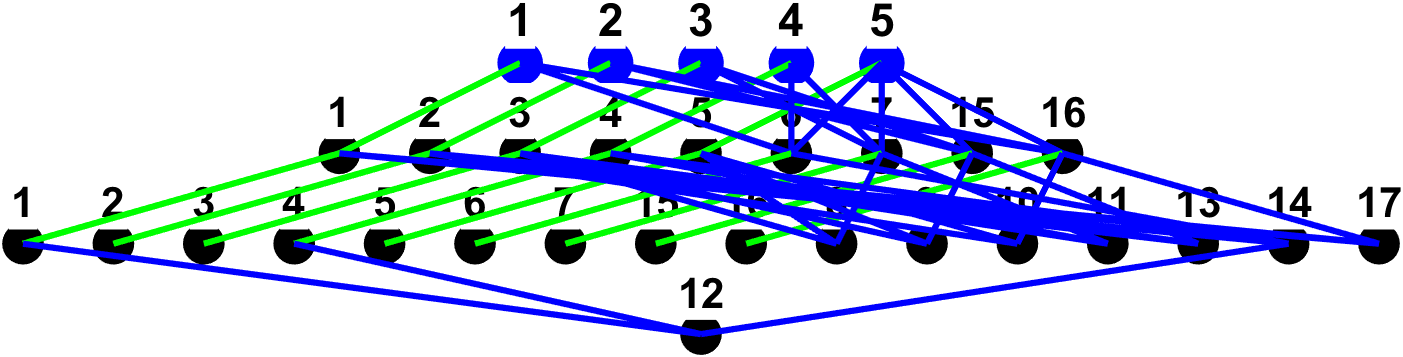}}
 \subfigure[]{\includegraphics[width=0.49\linewidth]{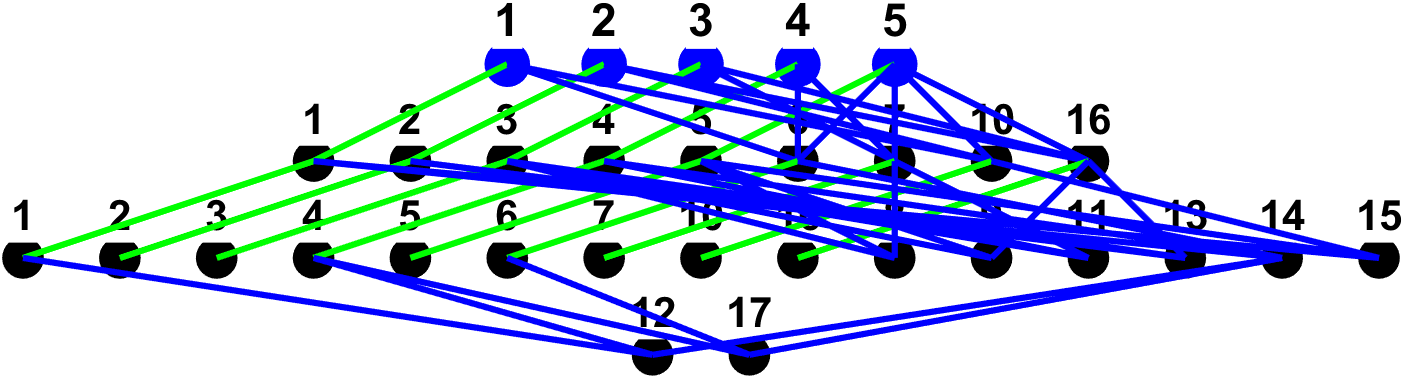}}
 \subfigure[]{\includegraphics[width=0.49\linewidth]{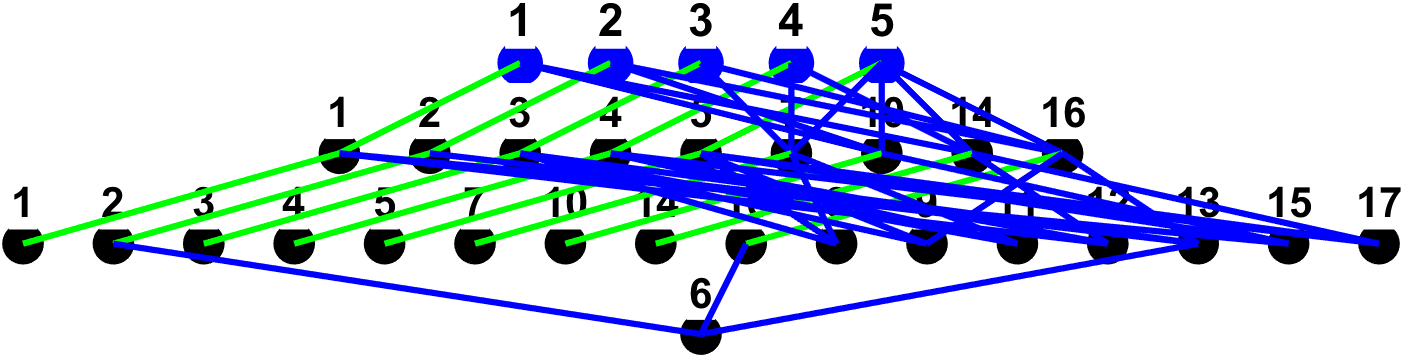}}
 \subfigure[]{\includegraphics[width=0.49\linewidth]{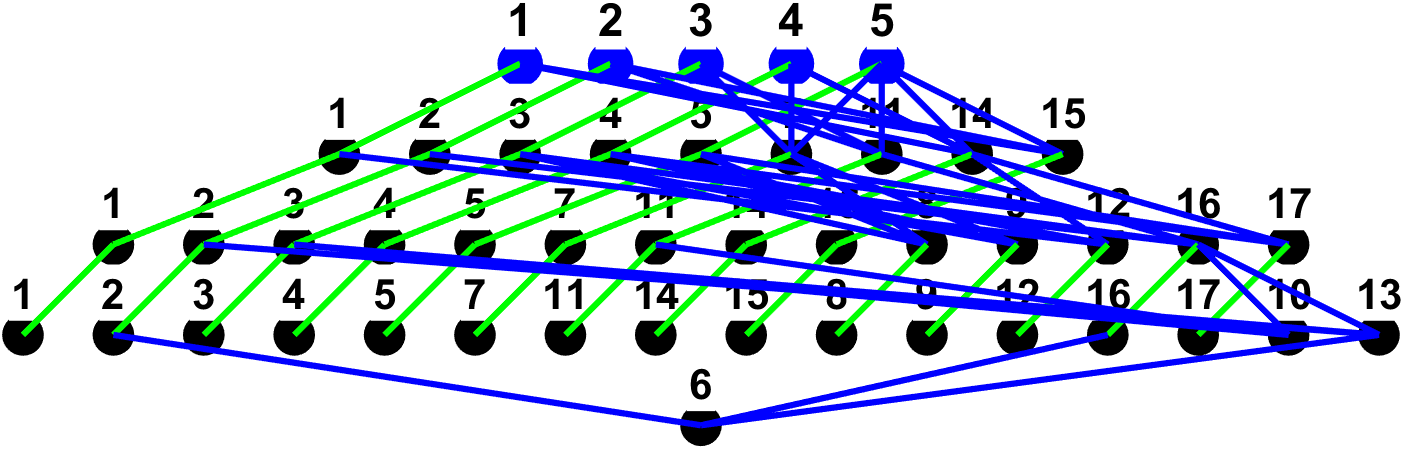}}
 \subfigure[]{\includegraphics[width=0.49\linewidth]{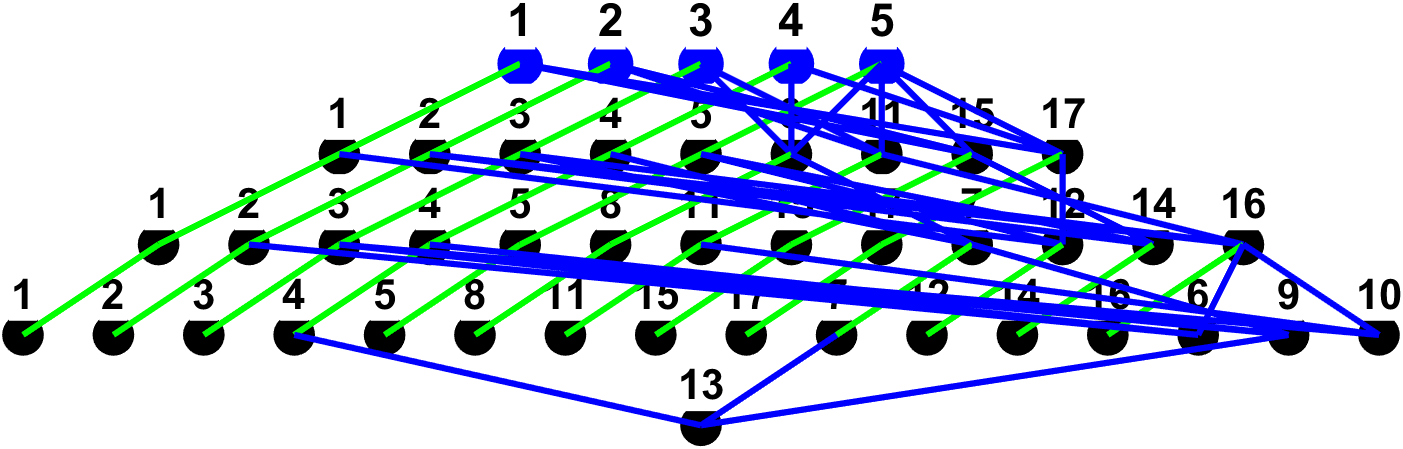}}
 \vspace{-0.4cm}
\caption{Mode-dependent DNN-based  communication
topologies used for decentralized surveillance. Each subfigure corresponds to
an operation mode
$\sigma\in\mathcal R$.
The directed edges define the information flow used by the active worker UAS
to compute their desired trajectories through recursive barycentric updates.
The communication graph remains fixed within each operation mode and is
reconfigured only when transitioning to the next mode.}
\label{dnnbasedcommunication}
\end{figure}
\section{Simulation Results}\label{Results}
We consider an multi-copter team consisting of $17$ agents indexed by
$\mathcal{V}=\{1,\ldots,17\}$, including five anchor agents
$\mathcal{A}=\{1,\ldots,5\}$ and twelve worker agents
$\mathcal{W}=\{6,\ldots,17\}$.
The workers are organized into $M=8$ operation modes that evolve cyclically according to the rotation graph
$\mathcal{G}_{\mathrm{rotation}}=(\mathcal{R},\mathcal{E}_{\mathcal{R}})$.
The desired coverage threshold is chosen as $x=1$.
The active-agent sets $\mathcal{V}_\sigma$ and the corresponding replaced-agent sets
$\mathcal{V}\setminus\mathcal{V}_\sigma$ for each operation mode
$\sigma\in\mathcal{R}$ are listed in Table~\ref{tab:rotation_modes}.
The communication topology is mode dependent, satisfying 
$
\mathcal{E}_\sigma\neq\mathcal{E}_{\sigma'}
4$ and $
\sigma\neq\sigma'$
although each  mode is repeated four times during one complete surveillance cycle. 


The proposed surveillance framework is evaluated over a rectangular region
\(
\mathcal{P}\subset\mathbb{R}^2
\)
of dimensions
\(52~\mathrm{m}\times52~\mathrm{m}\).
Each UAS is equipped with a downward-facing conical sensor whose footprint on
the ground is modeled as a circular sensing region of radius
\(2.5~\mathrm{m}\).
The surveillance domain is discretized into the node set
\(\mathcal D\), whose elements represent the candidate locations to be
searched. The reference positions of the anchor UAS are assumed to be known and remain
fixed throughout the mission. The initial operation mode is selected as
\(\sigma_0=1\), for which the unsearched node set satisfies
$
\mathcal U_1=\mathcal D$.
The initial reference positions of the worker UAS are then obtained using the
proposed Kullback-Leibler (KL) divergence-based initialization strategy,
which distributes the worker agents according to the spatial probability
distribution of the surveillance targets while accounting for the fixed anchor
configuration. Figure~\ref{KL_initial} illustrates the resulting initial
reference configuration, where the surveillance nodes are shown in black and
the worker reference positions are shown in red. For each operation mode
\(\sigma\in\mathcal R\),
the proposed framework computes the reference configuration of the subsequent
mode
\(\sigma^{+}\).
The resulting configuration shown in Fig. \ref{referenceall} specifies the reference positions of all active
worker UAS participating in the next surveillance cycle together with the
locations of the replacement agents determined by the proposed initialization
algorithm. The reference configuration is synthesized to maximize the
probability mass of the remaining unsearched target distributions while
preserving the communication topology and continuum-deformation constraints.
The unsearched target nodes associated with each mode are shown in black,
whereas the replacement-agent reference positions are shown in red. During each mode
$\sigma\in\mathcal R$,
the active worker UAS execute the surveillance task in a fully decentralized
manner using the proposed DNN-based communication
network. Each worker computes its desired trajectory solely from the states of
its prescribed in-neighbor agents, thereby eliminating the need for global
position information or centralized coordination. The mode-dependent DNN
communication topology employed throughout the surveillance mission is
illustrated in Fig.~\ref{dnnbasedcommunication}.

To illustrate the motion of an individual worker agent during the sequential coverage and replacement process, Fig.~\ref{fig:agent16_path} presents the complete planar path of agent~16. The trajectory is divided into the successive motion phases generated by the coverage-planning procedure. Ordinary motion segments are shown in blue, whereas the phases in which agent~16 belongs to the rotating replacement group are highlighted in red. 
The component-wise tracking responses are shown in Fig.~\ref{fig:agent16_components}. The upper and lower panels present the actual x- and y-coordinates of agent~16, respectively. 
\begin{figure}[ht]
\center
\includegraphics[width=3.3 in]{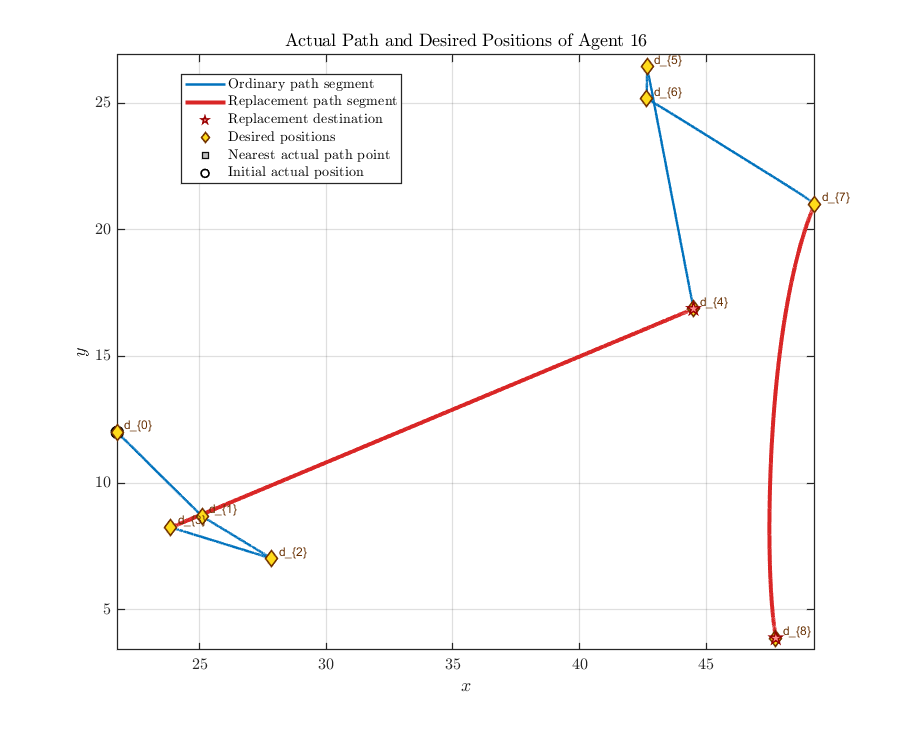}
\vspace{-.3cm}
\caption{%
    Actual planar trajectory of agent~16 during the sequential coverage
    operation. Ordinary motion phases are shown in blue, while the phases
    in which agent~16 is assigned to the rotating replacement group are
    highlighted in red. The markers indicate the desired positions
    generated at the successive planning phases. The proximity of the
    desired positions to the actual trajectory demonstrates phase-wise
    convergence of the fourth-order tracking controller.}
    \label{fig:agent16_path}
\end{figure} 
\begin{figure}[ht]
\center
\includegraphics[width=3.3 in]{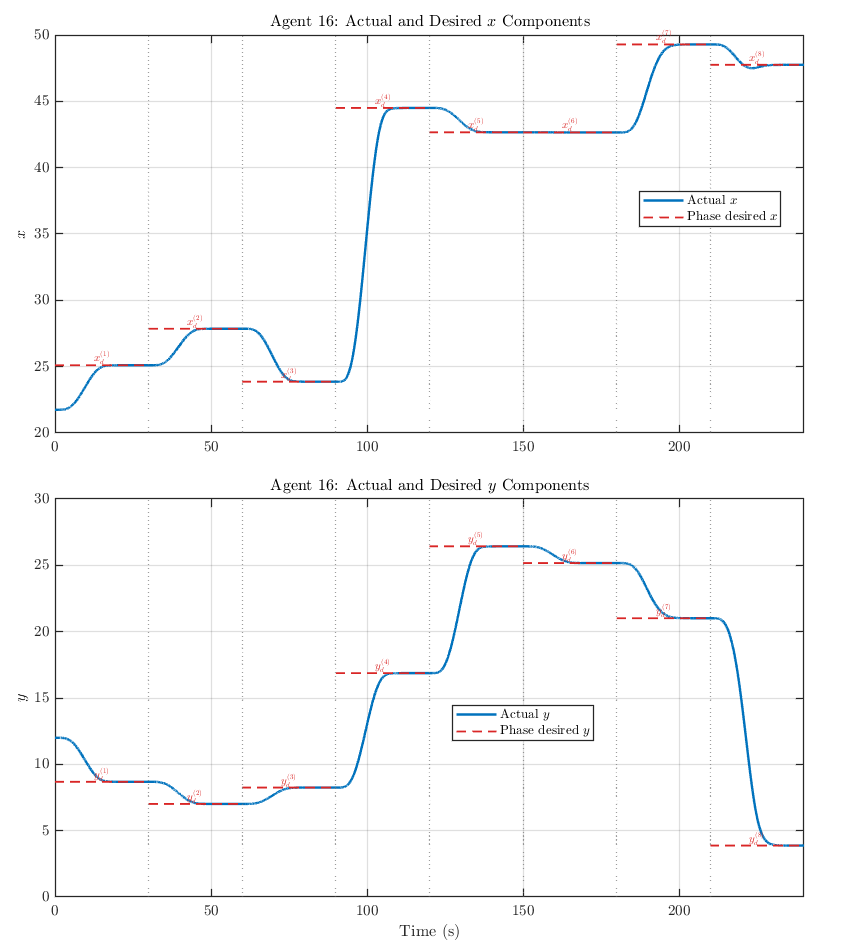}
\vspace{-.4cm}
\caption{%
    Component-wise tracking response of agent~16 over the complete
    multi-phase operation. The upper and lower panels show the actual
    \(x\)- and \(y\)-coordinates, respectively. The dashed horizontal
    segments indicate the terminal desired coordinate assigned during
    each phase, and the vertical dotted lines separate consecutive
    phases. The responses approach the corresponding phase-dependent
    desired coordinates while remaining continuous across the phase
    transitions.}
    \label{fig:agent16_components}
\end{figure}


\section{Conclusion}\label{Conclusion}

This paper presented a hybrid control framework for persistent aerial surveillance using energy-constrained multi-agent systems operating under deterministic cyclic team reconfiguration. The surveillance mission was formulated as a hybrid dynamical system that integrates continuous multi-agent coordination with discrete operation-mode transitions. A DNN-inspired communication architecture was developed to construct decentralized inter-agent communication topologies from mode-dependent reference configurations, enabling distributed computation of desired trajectories using only local information. A hierarchical linear temporal logic specification was introduced to formally characterize mission requirements, including anchor invariance, reference consistency, cyclic team rotation, distributed reachability, trajectory tracking, and prescribed coverage guarantees. Furthermore, an information-theoretic optimization framework was proposed to synthesize deterministic reference configurations that maximize surveillance effectiveness while preserving a mode-dependent communication topology.

\bibliographystyle{IEEEtran}
\bibliography{reference}
\end{document}